\documentclass{llncs}
\usepackage{tikz}
\tikzset{v/.style={fill, circle, inner sep=1pt, minimum size=3mm}}
\tikzset{e/.style={draw, line width=0.3mm}}
\tikzset{a/.style={draw,  ->, >=stealth, line width=0.6mm, color=blue!90!black}}
\tikzset{vwhite/.style={v, fill=white, draw, line width=0.3mm}}
\tikzset{n/.style={fill, square, inner sep=1pt, minimum size=3mm}}

\usepackage{amsmath}
\usepackage{gensymb}
\usepackage{float}
\usepackage[caption=false]{subfig}
\usepackage{enumitem}

\spnewtheorem{observation}{Observation}{\bfseries}{}
\begin{document}

\title{Orthogonal Compaction Using Additional Bends}

\author{ Michael J\"unger\inst{1} \and Petra Mutzel\inst{2} \and Christiane Spisla\inst{2}}

\institute{University of Cologne, Cologne, Germany\\
 \email{mjuenger@informatik.uni-koeln.de}
\and TU Dortmund University, Dortmund, Germany\\
\email{\{petra.mutzel,christiane.spisla\}@cs.tu-dortmund.de}
}
\maketitle

\begin{abstract}Compacting orthogonal drawings is a challenging task. Usually algorithms try to compute drawings with small area or edge length while preserving the underlying orthogonal shape. We present a one-dimensional compaction algorithm that alters the orthogonal shape of edges for better geometric results.
An experimental evaluation shows that we were able to reduce the total edge length and the drawing area, but at the expense of additional bends. 
 
 \end{abstract}

\section{Introduction}

The compaction problem in orthogonal graph drawing deals with constructing an area-efficient drawing on the orthogonal grid. Every edge is drawn as a sequence of horizonal and vertical segments, where the vertices and bends are placed on grid points.
Compaction has been studied in the context of the \emph{topology-shape-metrics approach}~\cite{TPM}. Here, in a first phase a combinatorial embedding is computed that determines the topology of the layout with the goal to minimize the number of crossings. In the second phase, a dimensionless orthogonal shape of the graph is determined by fixing the angles between adjacent edges and the bends along the edges. The goal is to minimize the number of bends.
In the third phase, metrics are added to the orthogonal shape.
In this context, first the coordinates of vertices and bends are assigned to grid points so that the given orthogonal shape is maintained. Finally, the (orthogonal) compaction problem asks for a drawing minimizing geometric aestetic criteria, such as the area of the drawing or the total edge length. The shape is not allowed to change.

Since the orthogonal compaction problem is NP-hard~\cite{Patrignani:1999:COC:645932.673201}, in practice heuristics are used that fix the $x$- (or $y$-, resp.) coordinates and solve the resulting compaction problem in one dimension. Given an initial drawing, the one-dimensional compaction problem with the goal of minimizing the height (or width, resp.) of the layout can be transformed to the longest path problem in a directed acyclic graph. If in addition the total edge length shall be minimized, the problem can be solved by computing a minimum cost flow.  
\begin{figure}[thb]

\subfloat[]{
\begin{minipage}{0.45\textwidth}
	\centering
\includegraphics[scale=0.09]{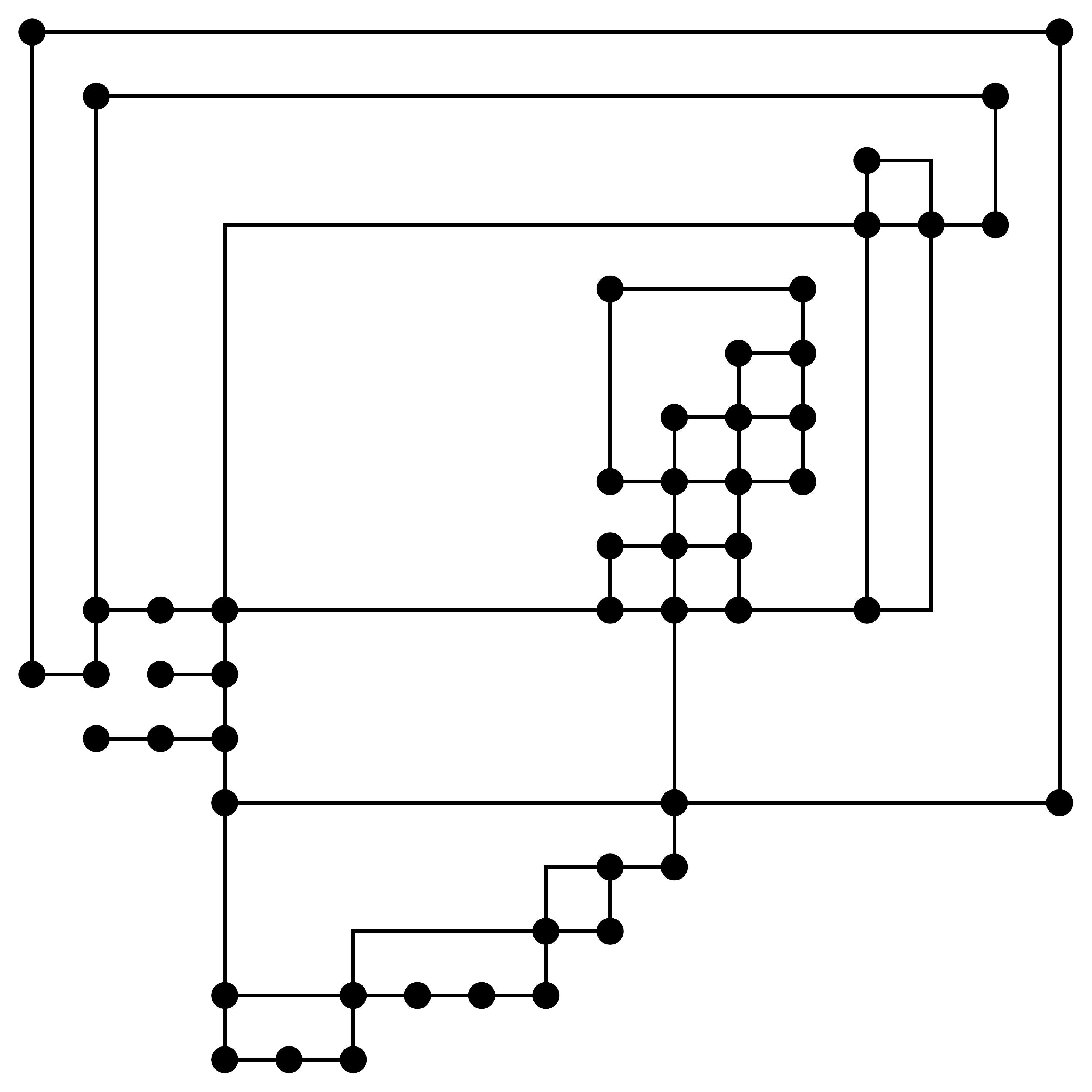}
\end{minipage}
}%end subfloat
\hfill
\subfloat[]{
\begin{minipage}{0.45\textwidth}	
	\centering
\includegraphics[scale=0.09]{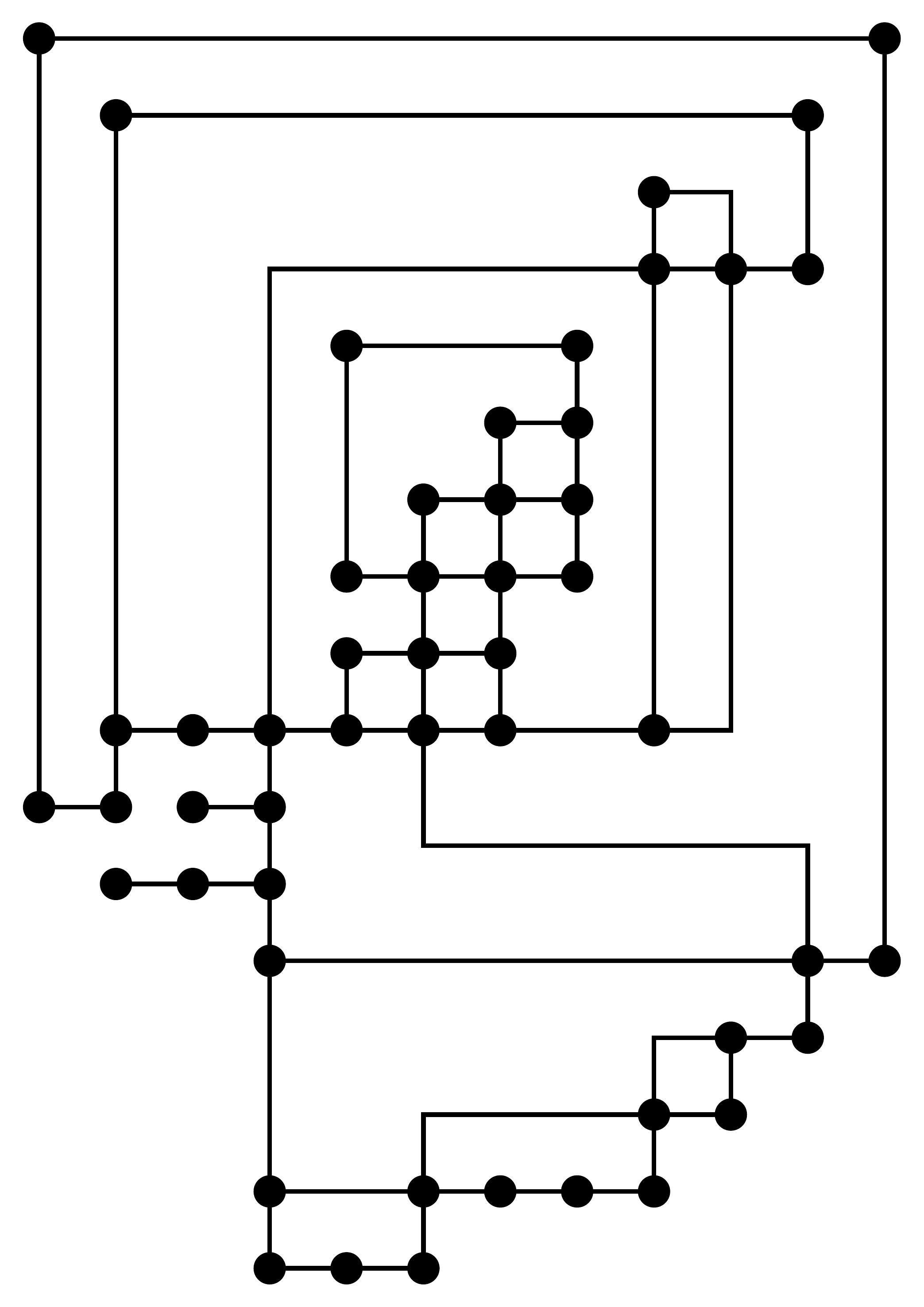}
\end{minipage}
}%end subfloat
\caption{(a) A drawing with large areas of white space due to shape restrictions. (b) Introducing two additional bends to one edge of the drawing leads to a smaller drawing.}
\label{erstesBeispiel}
\end{figure}
The topology-shape-metrics approach aims at drawings with a small number of crossings, a small number of bends, and a small drawing area. These goals are addressed in this order. And indeed, compared with other drawing methods, the number of crossings and bends is relatively small~\cite{DIBATTISTA1997303}. However, the layouts often contain large areas of white space. It seems that the goal of getting a small drawing area has not been achieved so far.
Consider the drawing in Fig.~\ref{erstesBeispiel}(a) which contains large areas of white space due to shape restrictions. 
By introducing two bends on one of the edges the drawing area can be reduced drastically.
This motivates us to study a compaction problem in which the shape conditions are relaxed.

This brings us back to the origin of the orthogonal compaction problem in VLSI-design (see, e.g.,~\cite{Lengauer:1990:CAI:92429}).
In contrast to the compaction problem considered in graph drawing, even the permutation of wires along the boundary of a component (and hence, changing the embedding) is allowed.

We suggest a moderate relaxation of the orthogonal compaction problem. 
More precisely, we suggest to study 
the \emph{one-dimensional monotone flexible edge compaction problem with fixed vertex star geometry}, henceforth the \emph{Fled-Five compaction  problem}, which asks for the minimization of the vertical (horizontal, resp.) edge length and allows changing the orthogonal shape of the edge, but preserves the $x$-monotonicity ($y$-monotonicity, resp.) of edge segments and prohibits changing the directions of the initial edge segments around the vertices. We present a polynomial-time algorithm based on a network flow model that solves the Fled-Five compaction problem to optimality. Our computational results show that repeated application of Fled-Five compaction in $x$- and $y$-direction is able to reduce the total edge length and the drawing area at the expense of additional bends.
 
This paper is organized as follows. We recall the state of the art in Sect.~\ref{stateOfTheArt} and some basic definition about orthogonal graph drawing and especially the  compaction phase in Sect.~\ref{Notation}. We present our new algorithm in Sect.~\ref{algorithm} and evaluate it experimentally in Sect.~\ref{test}.

\section{State-of-the-Art}
\label{stateOfTheArt}

Patrignani~\cite{Patrignani:1999:COC:645932.673201} has shown that planar orthogonal compaction is in general NP-hard and Bannister et al.~\cite{DBLP:journals/jgaa/BannisterES12} gave inapproximability results for the nonplanar case. But for some special cases there exist polynomial algorithms, e.g. if all faces are of rectangular shape~\cite{DBLP:books/ph/BattistaETT99} or if all faces are so-called \emph{turn-regular}~\cite{BRIDGEMAN200053} or have a unique completion~\cite{DBLP:conf/ipco/KlauM99}. Klau et al.~\cite{DBLP:conf/ipco/KlauM99} suggested a branch-and-cut approach to solve an integer linear program based on extending a pair of constraint graphs. 
However, in practice, heuristics are used which iteratively fix the $x$-, and then the $y$-coordinates, and solve the resulting one-dimensional compaction problem. This process is repeated until no further progress is made. One-dimensional compaction algorithms often use either network flow techniques or a longest path method in order to assign integer coordinates to the vertices, see e.g.~\cite{DBLP:conf/dagstuhl/1999dg} for an overview. An experimental comparison of planar compaction algorithms was presented by Klau et al.~\cite{Klau2001}.

Although there has been done some work to improve the quality of a drawing by changing its shape, e.g.~\cite{4M}, most compaction algorithms take as input an orthogonal representation and try to produce compact drawings with respect to that representation. This can lead to an unnecessarily large drawing area with unused space. Often better results in terms of area and edge length can be achieved if the orthogonal shape can be altered, as we have seen in Fig.~\ref{erstesBeispiel}. On the other hand, it might be desirable to not change a given drawing too much in order to preserve the mental map.

\section{Notation and Preliminary Results}
\label{Notation}

In this section we give basic definitions and notations. For more details on orthogonal drawings and graph drawing in general see e.g.~\cite{DBLP:books/ph/BattistaETT99},~\cite{DBLP:conf/dagstuhl/1999dg} or~\cite{DBLP:reference/crc/2013gd}.

\subsection{Orthogonal Graph Drawing}
For the rest of this paper we restrict ourselves to undirected \emph{4-graphs}, i.e. graphs whose vertices have at most four incident edges.
A graph $G=(V,E)$ with $|V|=n$ and $|E|=m$ is called \emph{planar} if it admits a drawing $\Gamma$ in the plane without edge crossings. Such a planar drawing of $G$ induces a \emph{(planar) embedding}, which is represented by a
circular ordered list of bordering edges for every \emph{face}. The unbounded region of a planar drawing is called \emph{external face}.

An \emph{orthogonal representation} $H$ is an extension of a planar embedding that gives combinatorial information about the orthogonal shape of a drawing. For every edge we provide information about the bends encountered while traversing the edge and the angle formed at vertices by two consecutive edges. If one of these angles is 270\degree\,or 360\degree\ we associate with $v$ a \emph{reflex corner}. An orthogonal representation is called \emph{normalized} if it has no bends.
An \emph{orthogonal grid drawing} $\Gamma$ of $G$ is a drawing in which every edge is drawn as a sequence of horizontal and vertical \emph{edge segments} and every vertex and bend has integer coordinates. Such a drawing induces an orthogonal representation $H_\Gamma$ and a \emph{star geometry} for every vertex fixing the directions of the initial line segments of its incident edges. If an edge first turns to the right and then to the left, or vice versa, we call this a \emph{double bend} and the edge segment between those two bends a \emph{middle segment}.
Every orthogonal representation can be normalized by replacing all bends in $H_\Gamma$ with dummy vertices of degree two, thus adding vertices to $G$ and $\Gamma$.

Since our new approach is based on a network flow model for one-dimensional compaction, we will give a brief introduction to minimum cost flows here. For more information about network flows, see~\cite{Ahuja:1993:NFT:137406}.
Let $N=(V_N,E_N)$ be a directed graph. Whenever we talk about flows we will call $N$ a \emph{network}, the members of $V_N$ \emph{nodes} and the members of $E_N$ \emph{arcs} (in contrast to vertices and edges in an undirected graph). Every arc $a$ has a \emph{lower bound} $l(a)\in \bbbr^{\geq 0}$, an \emph{upper bound} $u(a) \in \bbbr^{\geq 0} \cup \{\infty\}$ and a nonnegative \emph{cost} $c(a)$. A \emph{demand} $b(n) \in \bbbr$ is associated with every node. We call a function $x: A \rightarrow \bbbr^{\geq 0}$ a \emph{flow} if $x$ satisfies the following conditions:
\begin{align}
&l(a) \leq x(a) \leq u(a) &\forall\; a\in E_N \quad &\mbox{(capacity constraint)}\\
&\sum_{a=(k,l)}x(a) \;- \sum_{a=(j,k)} x(a) = b(k)  &\forall\; k \in V_N \quad &\mbox{(flow conservation)}
\end{align}
A \emph{minimum cost flow} is a flow $x$ with minimum total cost $c_x=\sum_{a\in E_N} x(a)c(a)$ under all feasible flows. The minimum cost flows we are interested in can be computed in $\mathcal{O}(|V_N|^{3/2}\log{}|V_N|)$ time~\cite{DBLP:journals/jgaa/CornelsenK12}.

\subsection{Compaction of Orthogonal Drawings}
We focus on the \emph{vertical (orthogonal) compaction problem} that receives as input a planar grid drawing $\Gamma$ of a graph $G$ with an orthogonal representation $H_\Gamma$, and asks for another planar orthogonal drawing $\Gamma^\prime$ of $G$ realizing $H_\Gamma$ so that the vertical edge length is minimized subject to fixed $x$-coordinates.
In the following we describe a flow-based method for vertical compaction similar to the coordinate assignment algorithm in Di Battista et al.~\cite{DBLP:books/ph/BattistaETT99}.

\label{improvComp}

Assume we have an initial grid drawing $\Gamma$. In a first step we normalize $H_\Gamma$ resulting in $\overline{\Gamma}$ and $\overline{H}_\Gamma$. Then we add vertical \emph{visibility edges} (so-called \emph{dissecting edges}).
We insert a vertical edge connecting each reflex corner with the vertex or edge that is visible in vertical direction, possibly introducing dummy vertices. This gives us $\widetilde{\Gamma}$ and $\widetilde{H}_\Gamma$.
This way we get rid of all reflex corners and have a representation with rectangular faces. 
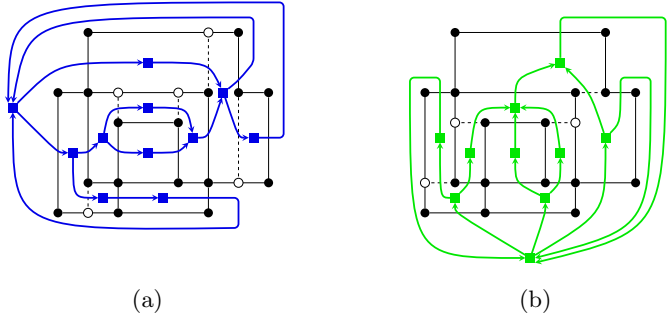
\begin{figure}[thb]
\begin{center}
\subfloat[]{
\resizebox{0.4\width}{!}{

	\begin{tikzpicture} %N_y
\tikzset{v/.style={fill, circle, inner sep=1pt, minimum size=3mm}}
\tikzset{e/.style={draw, line width=0.3mm}}
\tikzset{a/.style={draw,  ->, >=stealth, line width=0.6mm, color=blue!90!black}}
\tikzset{vwhite/.style={v, fill=white, draw, line width=0.3mm}}
\tikzset{n/.style={color=blue!90!black, fill, rectangle, inner sep=1pt, minimum size=3mm}}
	
	%G
	\node[v] at (0,0) (a) {};
	\node[v] at (2,0) (b) {};
	\node[v] at (5,0) {};
	\node[v] at (5,1) (c) {};
	\node[v] at (7,1) (d) {};
	\node[v] at (7,4) (e) {};
	\node[v] at (6,4) (f) {};
	\node[v] at (6,6) (g) {};
	\node[v] at (1,6) (h) {};
	\node[v] at (1,4) (i) {};
	\node[v] at (0,4) {};
	
	\draw (a) -| (c) -| (e) -| (g) -| (i) -| (a);
	
	\node[v] at (1,1) (j) {};
	\node[v] at (2,1) {};
	\node[v] at (4,1) (k) {};
	\node[v] at (2,3) (l) {};
	\node[v] at (4,3) (m) {};
	\node[v] at (5,4) (n) {};
	
	\draw (j) -| (n) -| (j);
	\draw (b) |- (m) -- (k);
	
	\node[vwhite] at (2,4) (o) {};
	\draw[e, dashed] (o) -- (2,3);

	\draw[e, dashed] (1,0) -- (1,1);
	\node[vwhite] at (1,0) {};
	\draw[e, dashed] (4,4) -- (4,3);
	\node[vwhite] at (4,4) {};
	\draw[e, dashed] (5,6) -- (5,4);
	\node[vwhite] at (5,6) {};
	\draw[e, dashed] (6,1) -- (6,4);
	\node[vwhite] at (6,1) {};	
	
	%N_y
	\node[n] at (-1.5, 3.5) (n0) {};
	\node[n] at (0.5, 2) (n1) {};
	\node[n] at (1.5, 2.5) (n2) {};
	\node[n] at (3, 2) (n3) {};
	\node[n] at (3.5, 0.5) (n4) {};
	\node[n] at (4.5, 2.5) (n5) {};
	\node[n] at (5.5,4) (n6) {};
	\node[n] at (1.5,0.5) (n7) {};
	\node[n] at (3,3.5) (n8) {};
	\node[n] at (3,5) (n9) {};
	\node[n] at (6.5,2.5) (n10) {};
	
	\draw[a] (n0) .. controls(0,2) .. (n1);
	\draw[a] (n0.north east) .. controls(0,5) .. (n9);
	\draw[a] (n1) .. controls(0.5,0.5).. (n7);
	\draw[a] (n7) -- (n4);
	\draw[a] (n1) .. controls(1,2) .. (n2);
	\draw[a] (n2) .. controls(2,2) .. (n3);
	\draw[a] (n2) .. controls(1.7,3.5) .. (n8);
	\draw[a] (n8) .. controls(4.3,3.5) .. (n5);
	\draw[a] (n9) .. controls(5,5) .. (n6);
	\draw[a] (n6) .. controls(6,2.5) .. (n10);
	\draw[a] (n3) .. controls(4,2) .. (n5);
	\draw[a] (n5) .. controls(5,2.5) .. (n6);
	
	\draw[a, rounded corners] (n4) -- (6,0.5) -- (6,-0.5) .. controls(-1.5,-0.6) .. (n0);
	\draw[a, rounded corners] (n6) ..controls(6,4.4).. (6.5,5) -- (6.5,6.5) .. controls(-1,6.5) .. (n0.north);
	\draw[a, rounded corners] (n10) -- (7.5,2.5) -- (7.5,7) .. controls(-1.6,7) .. (n0.north west);
	
	\node[] at (3.5,-2){};
	\end{tikzpicture}
	
}
}% end subfloat
\hspace{12mm}
\subfloat[]{
\resizebox{0.4\width}{!}{

	\begin{tikzpicture}
%	\begin{scope}[xshift=10cm]
	\tikzset{n/.style={color=green!90!black, fill, rectangle, inner sep=1pt, minimum size=3mm}}
	\tikzset{a/.style={draw,  ->, >=stealth, line width=0.6mm, color=green!90!black}}
	
	%G
	\node[v] at (0,0) (a) {};
	\node[v] at (2,0) (b) {};
	\node[v] at (5,0) {};
	\node[v] at (5,1) (c) {};
	\node[v] at (7,1) (d) {};
	\node[v] at (7,4) (e) {};
	\node[v] at (6,4) (f) {};
	\node[v] at (6,6) (g) {};
	\node[v] at (1,6) (h) {};
	\node[v] at (1,4) (i) {};
	\node[v] at (0,4) {};
	
	\draw (a) -| (c) -| (e) -| (g) -| (i) -| (a);
	
	\node[v] at (1,1) (j) {};
	\node[v] at (2,1) {};
	\node[v] at (4,1) (k) {};
	\node[v] at (2,3) (l) {};
	\node[v] at (4,3) (m) {};
	\node[v] at (5,4) (n) {};
	
	\draw (j) -| (n) -| (j);
	\draw (b) |- (m) -- (k);
	
	\draw[e, dashed] (5,4) -- (6,4);
	\draw[e, dashed] (0,1) -- (1,1);
	\draw[e, dashed] (1,3) -- (2,3);
	\draw[e, dashed] (4,3) -- (5,3);
	\node[vwhite] at (0,1) {};
	\node[vwhite] at (1,3) {};
	\node[vwhite] at (5,3) {};

	%N_x
	\node[n] at (3.5, -1.5) (n0) {};
	\node[n] at (1, 0.5) (n1) {};
	\node[n] at (4, 0.5) (n2) {};
	\node[n] at (3, 2) (n3) {};
	\node[n] at (3, 3.5) (n4) {};
	\node[n] at (6,2.5) (n5) {};
	\node[n] at (4.5,5) (n6) {};
	\node[n] at (0.5, 2.5) (n7) {};
	\node[n] at (1.5,2) (n8) {};
	\node[n] at (4.5,2) (n9) {};
	
	\draw[a] (n0) .. controls(1,0) .. (n1);
	\draw[a] (n0) .. controls(4,0) ..(n2);
	\draw[a] (n1) .. controls(0.5,0.5) .. (n7);
	\draw[a, rounded corners] (n7) --(0.5,4.5)--(-0.5,4.5).. controls(-0.5,-1.5).. (n0);
	\draw[a] (n1) .. controls(1.5,1)  .. (n8);
	\draw[a] (n8) .. controls(1.7,3.5) .. (n4);
	\draw[a] (n2) .. controls(4.5, 1) .. (n9);
	\draw[a] (n9) ..  controls(4.3,3.5) .. (n4);
	\draw[a] (n2) .. controls(3,1)  .. (n3);
	\draw[a] (n3) -- (n4);
	\draw[a] (n0) .. controls(6,0) .. (n5);
	\draw[a] (n4) .. controls(3,4.5) .. (n6);
	\draw[a] (n5) .. controls(5.5,4) .. (n6);
	\draw[a, rounded corners] (n5) .. controls(6.5,3) ..(6.5,4.5) -- (7.5,4.5) .. controls(7.5,-0.5).. (n0.east);
	
	\draw[a, rounded corners] (n6) -- (4.5,6.5) -- (8,6.5) .. controls(8,-0.6) .. (n0.south east);
	
	\node[] at (3.5,-2){};
	
%	\end{scope}
	
\end{tikzpicture}

}%end resizebox 
}%end subfloat
\end{center}
\caption[]{The flow networks $N_y$ for vertical compaction (a) and $N_x$ for horizontal compaction.(b). White vertices are dummy vertices and dashed edges are dummy edges.}

\label{classicalFlow}
\end{figure}
Now we are ready to construct the network $N_y$ for vertical compaction. For each face $f$ in representation $\widetilde{H}_\Gamma$ we add a node $n_f$ to $N_y$ with demand $b(n_f)=0$ and for every vertical edge $e$ with left face $f_l$ and right face $f_r$ we insert an arc $a_e = (n_{f_l} , n_{f_r})$ with lower bound $l(a_e)=1$ and upper bound $u(a_e)=\infty$. If $e$ is a dummy edge $a_e$ gets zero cost, otherwise a cost of one. See Fig.~\ref{classicalFlow} for an example.

Suppose we have computed a feasible flow. Now a unit of flow corresponds to one unit of vertical edge length. The $x$-coordinates remain unchanged. The capacity constraint assures that every edge gets a minimum length of one and the flow conservation constraint guarantees that every face is drawn as a proper rectangle. 
Finally, the visibility edges and artificial vertices can be removed. 

\begin{lemma}
\label{classicalFlowOptimalFor1DIM}
The above algorithm solves the vertical compaction problem to optimality.
\end{lemma}
\begin{proof}
In the vertical compaction problem we have to preserve the vertical visibility properties of $\Gamma$ in order to avoid overlapping graph elements. This is achieved by the newly added visibility edges. Because of the one-to-one correspondence of vertical length of an edge segment in the resulting drawing $\Gamma^\prime$ and the amount of flow carried by a non-zero-cost arc $a_e \in N_y$, the result of the minimum cost flow gives us a minimal vertical length assignment.\qed
\end{proof}

\section{The Fled-Five Compaction Approach}
\label{algorithm}

In this section we study the following relaxation of the one-dimensional compaction problem called the 
\emph{monotone flexible edge compaction problem with fixed vertex star geometry} (\emph{Fled-Five compaction problem}). For vertical compaction it is defined as follows:
Given a planar grid drawing $\Gamma$ of a 4-graph $G$ with an orthogonal representation $H_\Gamma$, compute another planar grid drawing of $G$ that minimizes the vertical edge length subject to fixed $x$-coordinates in which all horizontal edge segments of $\Gamma$ are drawn $x$-monotonically the vertex star geometry for all vertices as well as the planar embedding is maintained. 

In contrast to the classical compaction problem, it is not required here to realize the entire orthogonal representation, but only the local geometric surrounding of each vertex. The fixed $x$-coordinates and the $x$-monotonicity prohibits the lengthening of the total horizontal edge length (see Fig.~\ref{modifications}).
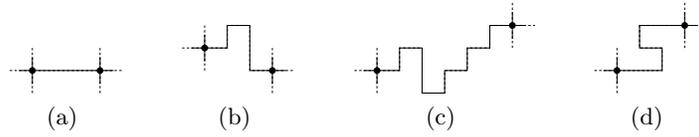
\begin{figure}[hbt]
\centering
\subfloat[]{
\resizebox{0.3\width}{!}{	
\begin{tikzpicture}
 	\node[v] at (1,0) {};
 	\node[v] at (4,0) {};
 	\draw[e] (0.4,0) -- (4.6,0);
 	\draw[e, dashed] (0,0) -- (5,0);
 	\draw[e, dashed] (1,-1) -- (1,1);
 	\draw[e, dashed] (4,-1) -- (4,1);
 	\draw[e] (1,-0.6) -- (1,0.6);
 	\draw[e] (4,-0.6) -- (4,0.6);
\end{tikzpicture} 	
}%end resize
}%end subfloat
\hspace{5mm}
\subfloat[]{
\resizebox{0.3\width}{!}{	
\begin{tikzpicture}
 	\node[v] at (1,0) {};
 	\node[v] at (4,-1) {};
 	\draw[e] (0.4,0) -| (2,1) -| (3,-1) --(4.6,-1);
 	\draw[e, dashed] (0,0) -| (2,1) -| (3,-1) --(5,-1);
 	\draw[e, dashed] (1,-1) -- (1,1);
 	\draw[e, dashed] (4,-2) -- (4,0);
 	\draw[e] (1,-0.6) -- (1,0.6);
 	\draw[e] (4,-1.6) -- (4,-0.4);
\end{tikzpicture} 	
}%end resize
}%end subfloat
\hspace{5mm}
\subfloat[]{
\resizebox{0.3\width}{!}{	
\begin{tikzpicture}
 	\node[v] at (1,0) {};
 	\node[v] at (7,2) {};
 	\draw[e] (0.4,0) -| (2,1) -| (3,-1) -|(4,0)-| (5,1)-| (6,2)--(7.6,2);
 	\draw[e, dashed] (0,0) -| (2,1) -| (3,-1) -|(4,0)-| (5,1)-| (6,2)--(8,2);
 	\draw[e, dashed] (1,-1) -- (1,1);
 	\draw[e, dashed] (7,1) -- (7,3);
 	\draw[e] (1,-0.6) -- (1,0.6);
 	\draw[e] (7,1.4) -- (7,2.6);
\end{tikzpicture} 	
}%end resize
}%end subfloat
\hspace{5mm}
\subfloat[]{
\resizebox{0.3\width}{!}{	
\begin{tikzpicture}
 	\node[v] at (1,0) {};
 	\node[v] at (4,2) {};
 	\draw[e] (0.4,0) -| (3,1) -| (2,2)-- (4.6,2);
 	\draw[e, dashed] (0,0) -| (3,1) -| (2,2)-- (5,2);
 	\draw[e, dashed] (1,-1) -- (1,1);
 	\draw[e, dashed] (4,1) -- (4,3);
 	\draw[e] (1,-0.6) -- (1,0.6);
 	\draw[e] (4,1.4) -- (4,2.6);
\end{tikzpicture} 	
}%end resize
}%end subfloat
\caption{Modifications in Fled-Five. (a) The original edge of length three. (b) Modification allowed in Fled-Five. (c) not allowed because of changed $x$-coordinate and (d) not allowed, because $x$-monotonicity is violated.}
\label{modifications}
\end{figure}

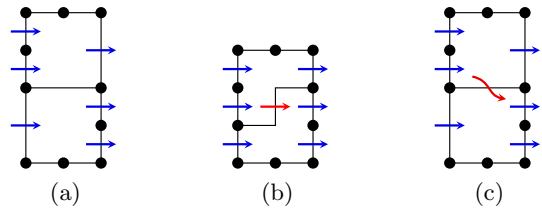
\begin{figure}[thb]
\begin{center}
\subfloat[]{
\resizebox{0.5\width}{!}{	
 		\begin{tikzpicture}
		\node[v] at (0,0) (a) {};
		\node[v] at (1,0) (b) {};
		\node[v] at (2,0) (c) {};
		\node[v] at (2,1) (d) {};
		\node[v] at (2,2) (e) {};
		\node[v] at (2,4) (f) {};
		\node[v] at (0,4) (g) {};
		\node[v] at (0,3) (h) {};
		\node[v] at (0,2) (i) {};
		\node[v] at (1,4) {};
		
		\draw[e] (0,0) to (2,0) to (2,4) to (0,4) to (0,0);
		\draw[e] (0,2) to (2,2);
		
		\draw[a] (-0.4, 1) to (0.4, 1);
		\draw[a] (-0.4, 3.5) to  (0.4, 3.5);
		\draw[a] (-0.4, 2.5) to  (0.4, 2.5);
		\draw[a] (1.6, 3) to  (2.4, 3);
		\draw[a] (1.6, 1.5) to  (2.4, 1.5);
		\draw[a] (1.6, 0.5) to  (2.4, 0.5);
	%	\node[] at (1,-1) {(a)};
	\end{tikzpicture}
 	}
}%end subfloat
	\hspace{9mm}
	\subfloat[]{
\resizebox{0.5\width}{!}{	
	\begin{tikzpicture}
		\node[v] at (0,0) (a) {};
		\node[v] at (1,0) (b) {};
		\node[v] at (2,0) (c) {};
		\node[v] at (2,1) (d) {};
		\node[v] at (2,2) (e) {};
		\node[v] at (2,3) (f) {};
		\node[v] at (0,3) (g) {};
		\node[v] at (0,2) (h) {};
		\node[v] at (0,1) (i) {};
		\node[v] at (1,3) {};
		
		\draw[e] (0,0) to (2,0) to (2,3) to (0,3) to (0,0);
		\draw[e] (0,1) to (1,1) to (1,2) to (2,2);
		
		\draw[a] (-0.4, 0.5) to (0.4, 0.5);
		\draw[a] (-0.4, 1.5) to  (0.4, 1.5);
		\draw[a] (-0.4, 2.5) to  (0.4, 2.5);
		\draw[a] (1.6, 2.5) to  (2.4, 2.5);
		\draw[a] (1.6, 1.5) to  (2.4, 1.5);
		\draw[a] (1.6, 0.5) to  (2.4, 0.5);
		\draw[a, color= red] (0.6,1.5) to (1.4,1.5);
		%\node[] at (1,-1) {(b)};
	\end{tikzpicture}
}
}%end subfloat
	\hspace{9mm}
\subfloat[]{
\resizebox{0.5\width}{!}{	
	\begin{tikzpicture}
		\node[v] at (0,0) (a) {};
		\node[v] at (1,0) (b) {};
		\node[v] at (2,0) (c) {};
		\node[v] at (2,1) (d) {};
		\node[v] at (2,2) (e) {};
		\node[v] at (2,4) (f) {};
		\node[v] at (0,4) (g) {};
		\node[v] at (0,3) (h) {};
		\node[v] at (0,2) (i) {};
		\node[v] at (1,4) {};
		
		\draw[e] (0,0) to (2,0) to (2,4) to (0,4) to (0,0);
		\draw[e] (0,2) to (2,2);
		
		\draw[a] (-0.4, 1) to (0.4, 1);
		\draw[a] (-0.4, 3.5) to  (0.4, 3.5);
		\draw[a] (-0.4, 2.5) to  (0.4, 2.5);
		\draw[a] (1.6, 3) to  (2.4, 3);
		\draw[a] (1.6, 1.5) to  (2.4, 1.5);
		\draw[a] (1.6, 0.5) to  (2.4, 0.5);
		\draw[a, color = red!90!black] (0.6,2.3) .. controls +(0.5,-0.1) and +(-0.5,0.1) .. (1.5,1.7);

		%\node[] at (1,-1) {(c)};
	\end{tikzpicture}
	
}
}%end subfloat
\end{center}	
	\caption[]{A graph and the arcs of the corresponding flow network (a,b) of the traditional algorithm and (c) in our algorithm}
	\label{exampleGraph}
\end{figure}
We adapt the network flow approach described in Sect.~\ref{improvComp}. But now we are able to introduce or remove double bends of certain edges in order to improve the vertical edge length.
We illustrate the concept by the following example. Figure~\ref{exampleGraph}(a) shows an optimally compacted drawing with respect to its orthogonal representation. The blue arcs show the arcs of the vertical compaction network (compare Fig.~\ref{classicalFlow}). For better readability we omitted the network nodes belonging to the faces. In Fig.~\ref{exampleGraph}(b), the same graph is shown, but with a different orthogonal representation. The new double bend in the middle edge leads to a smaller drawing. 
In this drawing it is possible to send flow between the two internal faces, since they are separated by a vertical edge segment. In the corresponding flow network there is a new network arc (red) connecting the upper and the lower face. 
This leads to the key idea of our approach. Introducing arcs in the vertical flow network between horizontally separated face nodes enables us to shift flow between them and therefore exchange vertical edge length at the expense of a double bend (see Fig.~\ref{exampleGraph}(c)). We can also reverse this thought. If there is an unnecessary double bend we can get rid of it by not sending flow over the arc corresponding to its middle segment and thus removing the middle segment from the drawing.

But we have to be careful here. First, we are compacting in vertical direction, so we cannot change the $x$-coordinates. If an edge has a double bend, it needs to have a horizontal expansion of at least two. Thus we will not consider  edges of length one as possible candidates for getting  a double bend.
Second, adding a double bend to $e$ introduces two reflex corners, one in both adjacent faces of~$e$. Two double bends may even cause an edge overlap, see Fig.~\ref{modifiedFlow}(a). We need to ensure that the computed flow corresponds to a feasible drawing. Therefore we will treat each grid point along an edge that could be part of a double bend as a potential reflex corner, which we will eliminate by dissecting.
We will now describe the algorithm in detail.
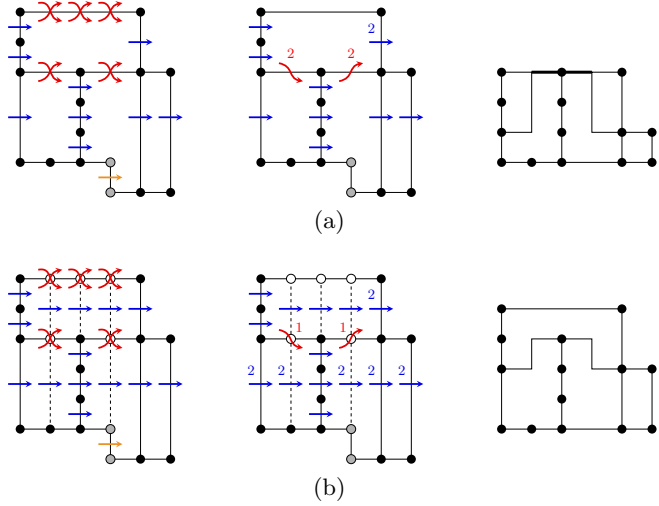
\begin{figure}[tbh]
\centering
\subfloat[]{
\resizebox{0.40\width}{!}{
	\begin{tikzpicture}

\begin{scope}[shift={(-8,0)}]
	\node[v] at (0,0) {};
	\node[v] at (0,3) {};
	\node[v] at (0,4) {};
	\node[v] at (0,5) {};
	\node[v] at (1,0) {};	
	\node[v] at (2,0) {};
	\node[v] at (2,1) {};
	\node[v] at (2,2) {};
	\node[v] at (2,3) {};
	%\node[v] at (3,0) {};
	\node[v] at (4,-1) {};
	\node[v] at (4,3) {};
	\node[v] at (4,5) {};
	\node[v] at (5,-1) {};
	\node[v] at (5,3) {};
	
	\draw[e] (0,0) -| (3,-1)-| (5,3) -| (0,5) -| (4,-1);
	\draw[e] (0,0) -- (0,3);
	\draw[e] (2,0) -- (2,3);
	
	%BendReplacer
	\node[v, fill=black!30, draw] at (3,0){};
	\node[v, fill=black!30,draw] at (3,-1){};
	
	%Jetzt Flusskanten
	\draw[a] (-0.4,1.5) -- (0.4,1.5);
	\draw[a] (-0.4,3.5) -- (0.4,3.5);
	\draw[a] (-0.4,4.5) -- (0.4,4.5);
	
	\draw[a] (1.6,0.5) -- (2.4,0.5);
	\draw[a] (1.6,1.5) -- (2.4,1.5);
	\draw[a] (1.6,2.5) -- (2.4,2.5);
	
	\draw[a] (3.6,1.5) -- (4.4,1.5);
	\draw[a] (3.6,4) -- (4.4,4);
	\draw[a] (4.6,1.5) -- (5.4,1.5);
	
	\draw[a, color=yellow!50!red!95!black] (2.6,-0.5) -- (3.4,-0.5);
	
	\draw[a, color = red!90!black] (0.6,3.3) .. controls +(0.5,0.1) and +(-0.5,0.1) .. (1.4,2.7);
	\draw[a, color=red!90!black] (0.6, 2.7) .. controls +(0.5,0.1) and +(-0.5,-0.1) .. (1.4,3.3);
	
	\draw[a, color = red!90!black] (2.6,3.3) .. controls +(0.5,0.1) and +(-0.5,0.1) .. (3.4,2.7);
	\draw[a, color = red!90!black] (2.6,2.7) .. controls +(0.5,0.1) and +(-0.5,-0.1) .. (3.4,3.3);
	
	\draw[a, color = red!90!black] (0.6,5.3) .. controls +(0.5,0.1) and +(-0.5,0.1) .. (1.4,4.7);
	\draw[a, color=red!90!black] (0.6, 4.7) .. controls +(0.5,0.1) and +(-0.5,-0.1) .. (1.4,5.3);
	
\draw[a, color = red!90!black] (1.6,5.3) .. controls +(0.5,0.1) and +(-0.5,0.1) .. (2.4,4.7);
	\draw[a, color = red!90!black] (1.6,4.7) .. controls +(0.5,0.1) and +(-0.5,-0.1) .. (2.4,5.3);	
	
	\draw[a, color = red!90!black] (2.6,5.3) .. controls +(0.5,0.1) and +(-0.5,0.1) .. (3.4,4.7);
	\draw[a, color = red!90!black] (2.6,4.7) .. controls +(0.5,0.1) and +(-0.5,-0.1) .. (3.4,5.3);
\end{scope}

	\node[v] at (0,0) {};
	\node[v] at (0,3) {};
	\node[v] at (0,4) {};
	\node[v] at (0,5) {};
	\node[v] at (1,0) {};	
	\node[v] at (2,0) {};
	\node[v] at (2,1) {};
	\node[v] at (2,2) {};
	\node[v] at (2,3) {};
	%\node[v] at (3,0) {};
	\node[v] at (4,-1) {};
	\node[v] at (4,3) {};
	\node[v] at (4,5) {};
	\node[v] at (5,-1) {};
	\node[v] at (5,3) {};
	
	\draw[e] (0,0) -|(3,-1)-| (5,3) -| (0,5) -| (4,-1);
	\draw[e] (0,0) -- (0,3);
	\draw[e] (2,0) -- (2,3);
	
	%BendReplacer
	\node[v, fill=black!30, draw] at (3,0){};
	\node[v, fill=black!30,draw] at (3,-1){};
	
	%Jetzt Flusskanten
	\draw[a] (-0.4,1.5) -- (0.4,1.5);
	\draw[a] (-0.4,3.5) -- (0.4,3.5);
	\draw[a] (-0.4,4.5) -- (0.4,4.5);
	
	\draw[a] (1.6,0.5) -- (2.4,0.5);
	\draw[a] (1.6,1.5) -- (2.4,1.5);
	\draw[a] (1.6,2.5) -- (2.4,2.5);
	
	\draw[a] (3.6,1.5) -- (4.4,1.5);
	\draw[a] (3.6,4) -- node[above=3pt, pos=0.1, scale=1.5]{2}(4.4,4);
	\draw[a] (4.6,1.5) -- (5.4,1.5);
	
	\draw[a, color = red!90!black] (0.6,3.3) .. controls +(0.5,-0.1) and +(-0.5,0.1) .. node[above=8pt, pos=0.5, scale=1.5, color=red!100] { \textbf{$2$}}(1.4,2.7);
	\draw[a, color = red!90!black] (2.6,2.7) .. controls +(0.5,0.1) and +(-0.5,-0.1) .. node[above=8pt, pos=0.5, scale=1.5, color=red!100] { \textbf{$2$}}(3.4,3.3);
	
\begin{scope}[shift={(8,0)}]
	\node[v] at (0,0) {};
	\node[v] at (0,1) {};
	\node[v] at (0,2) {};
	\node[v] at (0,3) {};
	\node[v] at (1,0) {};	
	\node[v] at (2,0) {};
	\node[v] at (2,1) {};
	\node[v] at (2,2) {};
	\node[v] at (2,3) {};
	%\node[v] at (3,0) {};
	\node[v] at (4,0) {};
	\node[v] at (4,1) {};
	\node[v] at (4,3) {};
	\node[v] at (5,0) {};
	\node[v] at (5,1) {};
	
	\draw[e] (0,0) -| (5,1) -| (3,3) -| (1,1) -| (0,3) -| (4,0);
	\draw[e] (0,0) -- (0,1);
	\draw[e] (2,0) -- (2,3);
	\draw[e, line width=1.0mm] (1,3) -- (3,3);
	
\end{scope}
\end{tikzpicture}
}%end resize
 } %end subfloat
 
  \subfloat[]{
\resizebox{0.40\width}{!}{
	\begin{tikzpicture}

\begin{scope}[shift={(-8,0)}]
	\node[v] at (0,0) {};
	\node[v] at (0,3) {};
	\node[v] at (0,4) {};
	\node[v] at (0,5) {};
	\node[v] at (1,0) {};	
	\node[v] at (2,0) {};
	\node[v] at (2,1) {};
	\node[v] at (2,2) {};
	\node[v] at (2,3) {};
	%\node[v] at (3,0) {};
	\node[v] at (4,-1) {};
	\node[v] at (4,3) {};
	\node[v] at (4,5) {};
	\node[v] at (5,-1) {};
	\node[v] at (5,3) {};
	
	\draw[e] (0,0) -|(3,-1)-| (5,3) -| (0,5) -| (4,-1);
	\draw[e] (0,0) -- (0,3);
	\draw[e] (2,0) -- (2,3);
	
	%Dummy edges
	\draw[e, dashed] (1,0) -- (1,5);
	\draw[e, dashed] (3,0) -- (3,5);
	\draw[e, dashed] (2,3) -- (2,5);
	
	%BendReplacer
	\node[v, fill=black!30, draw] at (3,0){};
	\node[v, fill=black!30,draw] at (3,-1){};

	%Dummy knoten
	\node[vwhite] at (1,3){};
	\node[vwhite] at (3,3){};
	\node[vwhite] at (1,5){};
	\node[vwhite] at (2,5){};
	\node[vwhite] at (3,5){};
	
	%Jetzt Flusskanten
	\draw[a] (-0.4,1.5) -- (0.4,1.5);
	\draw[a] (-0.4,3.5) -- (0.4,3.5);
	\draw[a] (-0.4,4.5) -- (0.4,4.5);
	
	\draw[a] (1.6,0.5) -- (2.4,0.5);
	\draw[a] (1.6,1.5) -- (2.4,1.5);
	\draw[a] (1.6,2.5) -- (2.4,2.5);
	
	\draw[a] (3.6,1.5) -- (4.4,1.5);
	\draw[a] (3.6,4) -- (4.4,4);
	\draw[a] (4.6,1.5) -- (5.4,1.5);
	
	\draw[a] (0.6, 1.5) -- (1.4,1.5);
	\draw[a] (2.6, 1.5) -- (3.4,1.5);
	
	\draw[a] (0.6, 4) -- (1.4,4);
	\draw[a] (1.6, 4) -- (2.4,4);
	\draw[a] (2.6, 4) -- (3.4,4);
	
	\draw[a, color=yellow!50!red!95!black] (2.6,-0.5) -- (3.4,-0.5);
	
	\draw[a, color = red!90!black] (0.6,3.3) .. controls +(0.5,0.1) and +(-0.5,0.1) .. (1.4,2.7);
	\draw[a, color=red!90!black] (0.6, 2.7) .. controls +(0.5,0.1) and +(-0.5,-0.1) .. (1.4,3.3);
	
	\draw[a, color = red!90!black] (2.6,3.3) .. controls +(0.5,0.1) and +(-0.5,0.1) .. (3.4,2.7);
	\draw[a, color = red!90!black] (2.6,2.7) .. controls +(0.5,0.1) and +(-0.5,-0.1) .. (3.4,3.3);
	
	\draw[a, color = red!90!black] (0.6,5.3) .. controls +(0.5,0.1) and +(-0.5,0.1) .. (1.4,4.7);
	\draw[a, color=red!90!black] (0.6, 4.7) .. controls +(0.5,0.1) and +(-0.5,-0.1) .. (1.4,5.3);
	
\draw[a, color = red!90!black] (1.6,5.3) .. controls +(0.5,0.1) and +(-0.5,0.1) .. (2.4,4.7);
	\draw[a, color = red!90!black] (1.6,4.7) .. controls +(0.5,0.1) and +(-0.5,-0.1) .. (2.4,5.3);	
	
	\draw[a, color = red!90!black] (2.6,5.3) .. controls +(0.5,0.1) and +(-0.5,0.1) .. (3.4,4.7);
	\draw[a, color = red!90!black] (2.6,4.7) .. controls +(0.5,0.1) and +(-0.5,-0.1) .. (3.4,5.3);
\end{scope}

	\node[v] at (0,0) {};
	\node[v] at (0,3) {};
	\node[v] at (0,4) {};
	\node[v] at (0,5) {};
	\node[v] at (1,0) {};	
	\node[v] at (2,0) {};
	\node[v] at (2,1) {};
	\node[v] at (2,2) {};
	\node[v] at (2,3) {};
	%\node[v] at (3,0) {};
	\node[v] at (4,-1) {};
	\node[v] at (4,3) {};
	\node[v] at (4,5) {};
	\node[v] at (5,-1) {};
	\node[v] at (5,3) {};

	\draw[e] (0,0) -|(3,-1)-| (5,3) -| (0,5) -| (4,-1);
	\draw[e] (0,0) -- (0,3);
	\draw[e] (2,0) -- (2,3);
	
	%Dummy edges
	\draw[e, dashed] (1,0) -- (1,5);
	\draw[e, dashed] (3,0) -- (3,5);
	\draw[e, dashed] (2,3) -- (2,5);
	
	%Dummy knoten
	\node[vwhite] at (1,3){};
	\node[vwhite] at (3,3){};
	\node[vwhite] at (1,5){};
	\node[vwhite] at (2,5){};
	\node[vwhite] at (3,5){};
	
	%BendReplacer
	\node[v, fill=black!30, draw] at (3,0){};
	\node[v, fill=black!30,draw] at (3,-1){};

	%Jetzt Flusskanten
	\draw[a] (-0.4,1.5) -- node[above=3pt, pos=0.1, scale=1.5]{2}(0.4,1.5);
	\draw[a] (-0.4,3.5) -- (0.4,3.5);
	\draw[a] (-0.4,4.5) -- (0.4,4.5);
	
	\draw[a] (0.6,1.5) -- node[above=3pt, pos=0.1, scale=1.5]{2}(1.4,1.5);
	\draw[a] (0.6,4) -- (1.4,4);
	
	\draw[a] (1.6,0.5) -- (2.4,0.5);
	\draw[a] (1.6,1.5) -- (2.4,1.5);
	\draw[a] (1.6,2.5) -- (2.4,2.5);
	\draw[a] (1.6,4) -- (2.4,4);
	
	\draw[a] (2.6,1.5) -- node[above=3pt, pos=0.1, scale=1.5]{2}(3.4,1.5);
	\draw[a] (2.6,4) -- (3.4,4);
	
	\draw[a] (3.6,1.5) -- node[above=3pt, pos=0.1, scale=1.5]{2}(4.4,1.5);
	\draw[a] (3.6,4) -- node[above=3pt, pos=0.1, scale=1.5]{2}(4.4,4);
	\draw[a] (4.6,1.5) -- node[above=3pt, pos=0.1, scale=1.5]{2}(5.4,1.5);
	
	\draw[a, color = red!90!black] (0.6,3.3) .. controls +(0.5,-0.1) and +(-0.5,0.1) .. node[above=8pt, pos=0.9, scale=1.5, color=red!100] { \textbf{$1$}}(1.4,2.7);
	\draw[a, color = red!90!black] (2.6,2.7) .. controls +(0.5,0.1) and +(-0.5,-0.1) .. node[above=8pt, pos=0.1, scale=1.5, color=red!100] { \textbf{$1$}}(3.4,3.3);
	
\begin{scope}[shift={(8,0)}]
	\node[v] at (0,0) {};
	\node[v] at (0,2) {};
	\node[v] at (0,3) {};
	\node[v] at (0,4) {};
	\node[v] at (1,0) {};	
	\node[v] at (2,0) {};
	\node[v] at (2,1) {};
	\node[v] at (2,2) {};
	\node[v] at (2,3) {};
	%\node[v] at (3,0) {};
	\node[v] at (4,0) {};
	\node[v] at (4,2) {};
	\node[v] at (4,4) {};
	\node[v] at (5,0) {};
	\node[v] at (5,2) {};
	
	\draw[e] (0,0) -| (5,2) -| (3,3) -| (1,2) -| (0,4) -| (4,0);
	\draw[e] (0,0) -- (0,2);
	\draw[e] (2,0) -- (2,3);

\end{scope}
\end{tikzpicture}
}%end resize
 } %end subfloat
\caption{Graph with modified flow network, a minimum cost flow (all unlabeled arcs carry one unit of flow) and the resulting drawing (a) without additional dissecting edges and (b) including them. The bold edge in (a) indicates an edge overlap and the grey vertices result from normalization, the white vertices in (b) are bend vertices.}
\label{modifiedFlow}
\end{figure}

Our algorithm works in three phases: dissection, computation of a minimum cost flow and transforming the flow into a new drawing.
First we normalize $H_\Gamma$ to $\overline{H}_\Gamma$.
For dissection we split the horizontal edges of $\overline{\Gamma}$ by placing an artificial \emph{bend vertex} on every inner grid point of an horizontal edge. The bend vertices may later be transformed to double bends.
Then we dissect $\overline{\Gamma}$ as described in Sect.~\ref{Notation} and treat every bend vertex as a reflex corner in its adjacent faces (see \ref{modifiedFlow}(b)). That means, we may dissect the drawing in vertical stripes of unit length. Doing so, we solve both of the problems mentioned above. Since we consider only edges of length at least two, we guarantee that the edge will be long enough for a double bend and by inserting the visibility edges we keep the vertical separation of graph elements intact. This results in the extensions $\widetilde{\Gamma}$ and $\widetilde{H}_\Gamma$ and $\widetilde{G}$.
\begin{observation}
\label{AnzahlKnoten}
After this dissection method the number of vertices and edges in $\widetilde{G}$ is $\mathcal{O}(A^2)$, where $A = h_\Gamma \cdot w_\Gamma$ and $h_\Gamma$ ($w_\Gamma$) is the height (width) of $\Gamma$.
\end{observation}

Next we construct the network from Sect.~\ref{improvComp}. Additionally to the already introduced arcs (blue arcs in Fig.~\ref{modifiedFlow}) we add two arcs $a^{\uparrow}_v$ and $a^\downarrow_v$ for every bend vertex (red curved arcs in Fig.~\ref{modifiedFlow}). Notice that every such bend vertex $v$ has four adjacent faces, one at the lower and upper left and right. If it lies on the external face,  two of the adjacent faces may coincide. Arc $a^{\uparrow}_v$ goes from the lower left face of $v$ to its upper right face and $a^{\downarrow}_v$ goes from the upper left to the lower right face. Flow on one of these arcs will lead to a double bend in the edge segment of $G$ that is split by $v$. The lower bound of these arcs is zero, the upper bound is set to infinity and the cost is one. For every edge $e \in \widetilde{G}$ that is a middle segment in $G$ we decrease the lower bound of its arc $a_e$ to zero, since this is a vertical edge segment, that we may delete (orange arc in Fig.~\ref{modifiedFlow}).

After computing a minimum cost flow in this network, we interpret the result in the following way: For every vertical edge $e$ of $\widetilde{G}$ we translate the amount of flow on the arc $a_e$ of $N_y$ to the length of $e$. Let $a^\uparrow_v$ be an arc corresponding to a bend vertex $v$ carrying $k$ units of flow. Let $e$ be the split horizontal edge and $x_v$ the $x$-coordinate of $v$. Then $e$ will start at its left endpoint in horizontal direction, bend downwards at $x$-coordinate $x_v$, proceed for $k$ units in $y$-direction and then continue to the right to its other endpoint. If we deal with an arc of the form $a^\downarrow_v$ the corresponding edge will do an upward bend at $x_v$. Let $e^\prime \in \widetilde{G}$ be an edge that corresponds to a middle segment and let $a_{e^\prime}$ be the corresponding arc. If $a_{e^\prime}$ carries no flow we do not assign any vertical length to $e^\prime$, hence the double bend to which $e^\prime$ belongs collapses. Notice that flow on %an arc $a^\uparrow_v$ or $a^\downarrow_v$ and thus flow on 
every non-zero-cost arc corresponds to vertical edge length. Finally, we remove all dummy edges and vertices. See Fig.~\ref{modifiedFlow} for an example.

Let us assume that the width and height of $\Gamma$ are bounded by the number of its vertices and bends. Otherwise there would be a grid line without any vertex or bend on it. We can delete such grid lines iteratively until we reach our bound.

\begin{theorem}
\label{mainTheorem}
Let $\Gamma$ be a planar grid drawing of a 4-graph $G$ with an orthogonal representation $H_\Gamma$. Let $\bar n$ be the number of vertices and bends of $\Gamma$. Then the above described extended network-based compaction algorithm takes $\mathcal{O} (\bar n^{3}\log{}\bar n)$ time and solves the vertical Fled-Five compaction problem to optimality.
\end{theorem}
\begin{proof}
Because of the visibility edges we will maintain visibility properties of~$\Gamma$. Every vertical edge segment of~$G$ that is not a middle segment gets a minimum length of one due to the lower bound of the corresponding arc in $N_y$. Every bend vertex can savely be turned to a double bend if its corresponding arc carries flow, since we add visibility edges to the top and bottom of it. By this modification we do not change the vertex star geometry of $\Gamma$, because bend vertices are not part of $G$. Every middle segment can be removed if there is no flow on the corresponding arc. Because its edge $e \in \Gamma$ has a horizontal expansion of at least two, $e$ will still have a proper length and due to the dissection phase visibility properties are maintained. This modification also does not affect the vertex star geometry of $\Gamma$, since a middle segment is not adjacent to a vertex of $G$. So every face will have a rectangular shape, no matter what modification we apply, and due to flow conservation every face and thus the entire graph is drawn consistently. Similar to Lemma~\ref{improvComp} the length of vertical segments of $\Gamma^\prime$ is equal to the cost of the computed minimum cost flow. 
The horizontal edge segments maintain their $x$-monotonicity, because we only add vertical segments to the drawing. Hence the horizontal edge lengths and the $x$-coordinates of the vertices  in $\Gamma$ stay the same. Because the minimum cost flow gives us a minimal vertical length assignment we have an optimal solution for the vertical Fled-Five compaction problem. 

For the running time we first consider the dissection phase. By Observation~\ref{AnzahlKnoten} we end up with $\mathcal{O}(\bar n^2)$ vertices. For inserting dummy edges based on visibility properties we can use a sweep-line algorithm that runs in $\mathcal{O}(\bar n^2\log{}\bar n)$ time in our case.
The construction of the flow network runs in linear time in the size of~$\widetilde{G}$. 
The network is planar and has $\mathcal{O}(\bar n^2)$ nodes. 
For planar flow networks with $n$ nodes and $\mathcal{O}(n)$ arcs there exists a $\mathcal{O}(n^{3/2}\log{} n)$ time algorithm for computing a minimum cost flow~\cite{DBLP:journals/jgaa/CornelsenK12}.
Therefore we have a total running time of $\mathcal{O} (\bar n^{3}\log{}\bar n)$.
\qed
\end{proof}

{%\bfseries Note} 
The cubic part of the running time comes from the possibly quadratic number of bend vertices and dissecting edges. So if we restrict the number of bend vertices to be linear, we can decrease the running time to $\mathcal{O} (\bar n^{3/2}\log{}\bar n)$.

\paragraph{Controlling the Number of New Bends.}
Although the above compaction approach may reduce the total edge length, the number of bends in the resulting drawing may increase. A possible approach to control the number of new bends could be to bound the number of specific arcs used by a feasible flow. This is known as the binary case of the \emph{budget-constrained minimum cost flow problem} for which Holzhauser et al.~\cite{Holzhauser:2016:BMC:2911577.2911605} showed strong NP-completeness.
In fact, in this model we have no other control over the number of additional bends than restricting the number of bend vertices. Each such vertex can generate two new bends.

\section{Experimental Evaluation}
\label{test}

In our experimental evaluation we compare the new compaction approach \textsf{FF} with the flow-based method \textsf{TRAD} described in Sect.~\ref{Notation}. Both lead to a planar drawing with optimal total edge length for the according one-dimensional compaction problem. 

The algorithm was implemented in C++ using the OGDF library~\cite{OGDF}. We have computed orthogonal drawings in the traditional way using the topology-shape-metrics approach, applying a constructive compaction method in order to get a feasible planar grid drawing with a normalized orthogonal representation~$H$. We have taken these drawings as input for both approaches \textsf{FF} and \textsf{TRAD} that repeatedly apply horizontal and vertical compaction steps until no further improvement can be achieved. Recall that although both approaches are optimal for the one-dimensional compaction problems, repeated application of horizontal and vertical compaction steps does in general not lead to drawings with optimal total edge length for the two-dimensional compaction problems. Notice also that due to the alternating repetitions in \textsf{FF} the monotonicity of edge segments may no longer be maintained.

We state the following hypotheses regarding the results of our new approach \textsf{FF} compared to those of \textsf{TRAD}:
\begin{enumerate}[leftmargin=10mm,itemsep=2mm]
\item[\textbf{(H1)}] The total edge length and therefore the area of the drawings will decrease. We will also examine the change of the maximum edge length. It is hard to predict the behaviour, because on the one hand adding a double bend lengthens an edge, but it could also lead to shorter edges in other places.
\item[\textbf{(H2)}] The number of bends will rise significantly.
\item[\textbf{(H3)}] Although there will be many more dissecting edges, the running time will not increase drastically.
\end{enumerate}

\begin{figure}[tbh]
\subfloat[]{
	\begin{minipage}{0.48\textwidth}
	\includegraphics[scale=0.4]{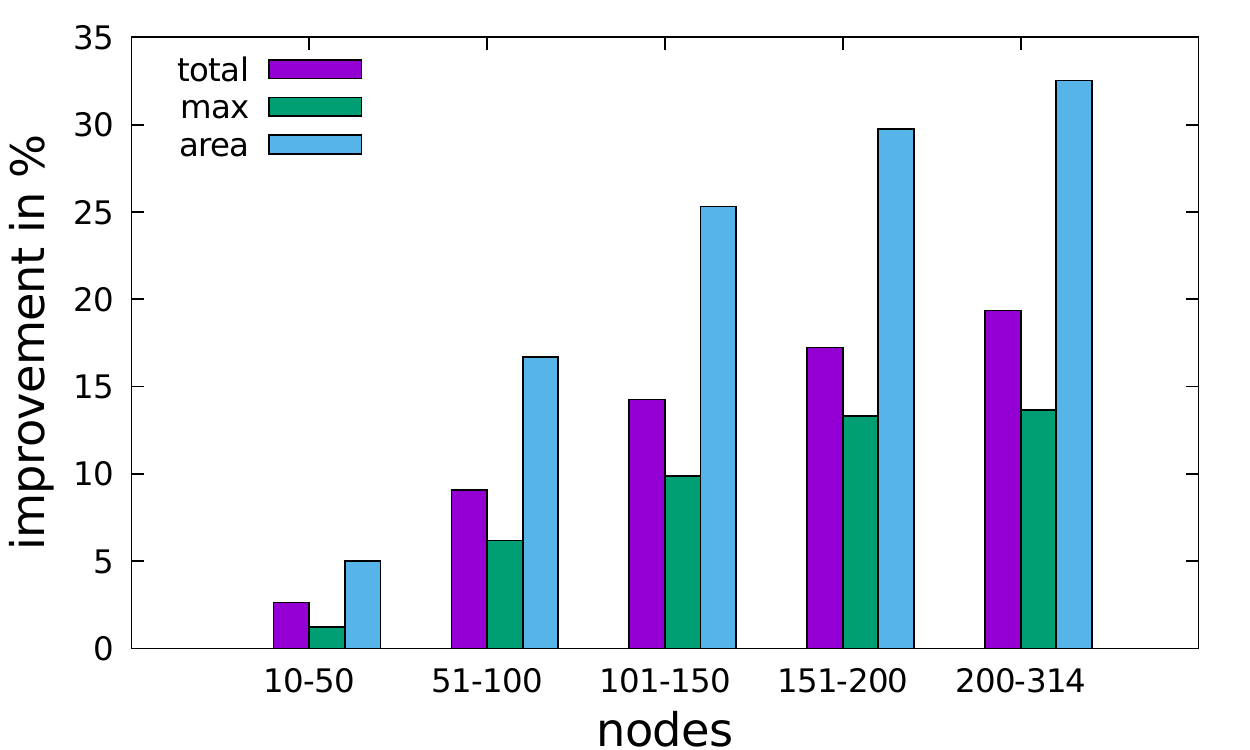}
	\end{minipage}
	\hfill
	\begin{minipage}{0.48\textwidth}
	\includegraphics[scale=0.4]{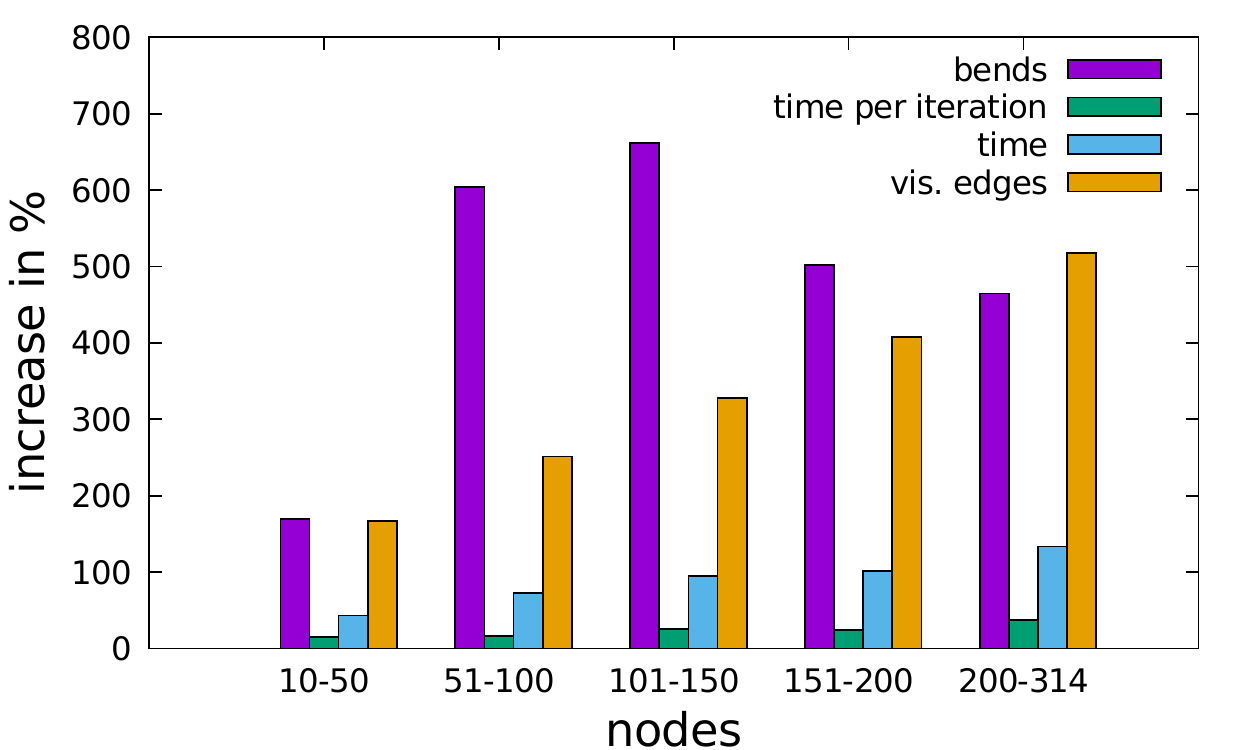}
	\end{minipage}
}%end subfloat

\subfloat[]{
	\begin{minipage}{0.48\textwidth}
	\includegraphics[scale=0.4]{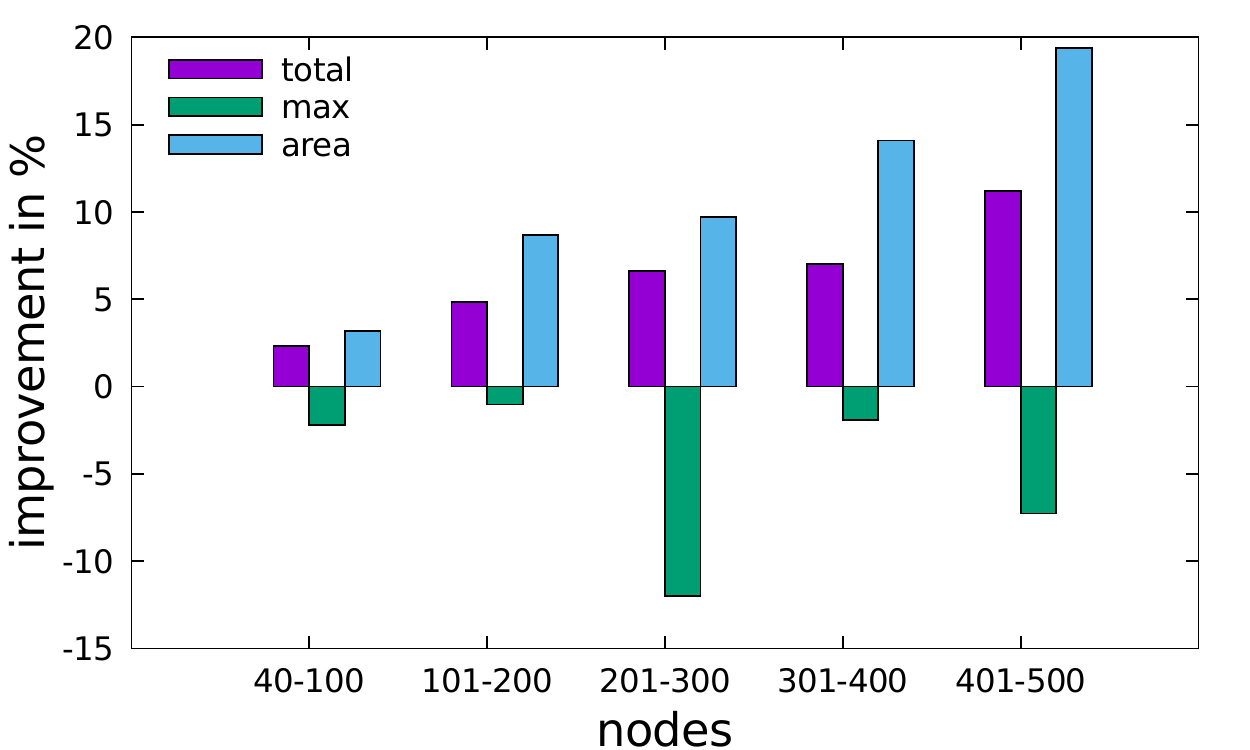}
	\end{minipage}
	\hfill
	\begin{minipage}{0.48\textwidth}
	\includegraphics[scale=0.4]{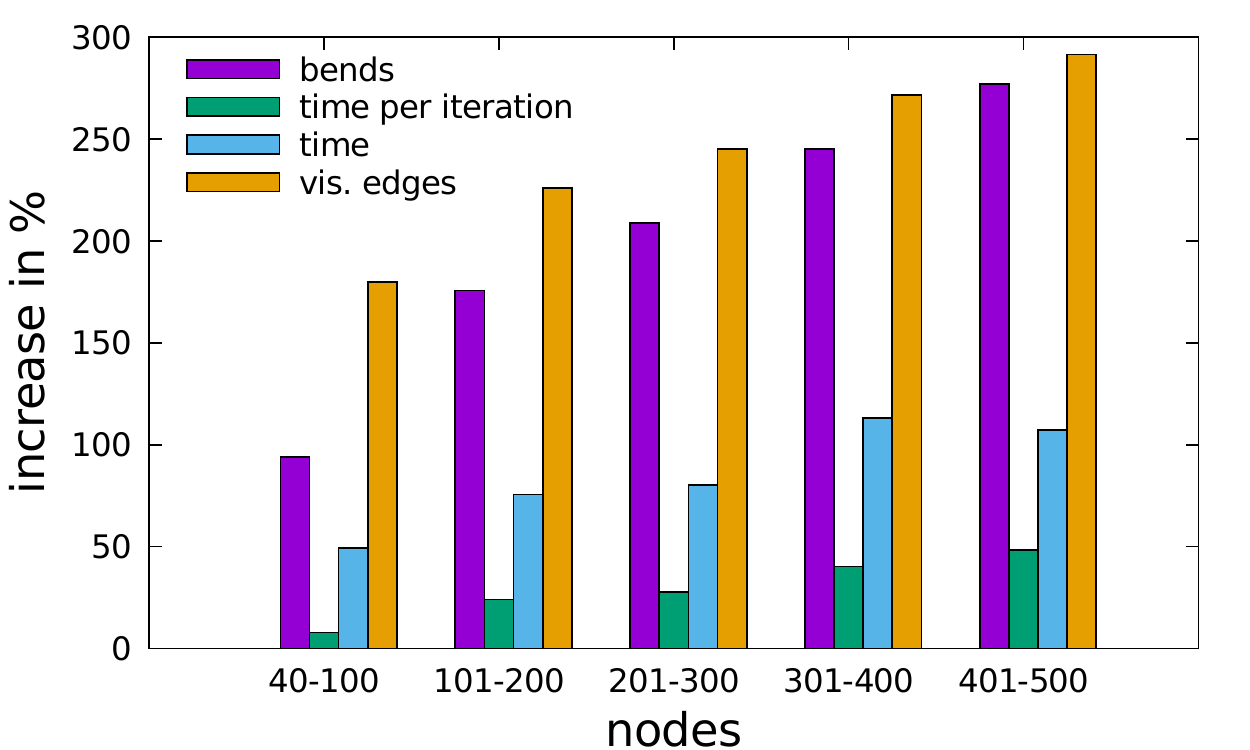}
	\end{minipage}
}%end subfloat

\subfloat[]{
	\begin{minipage}{0.48\textwidth}
	\includegraphics[scale=0.4]{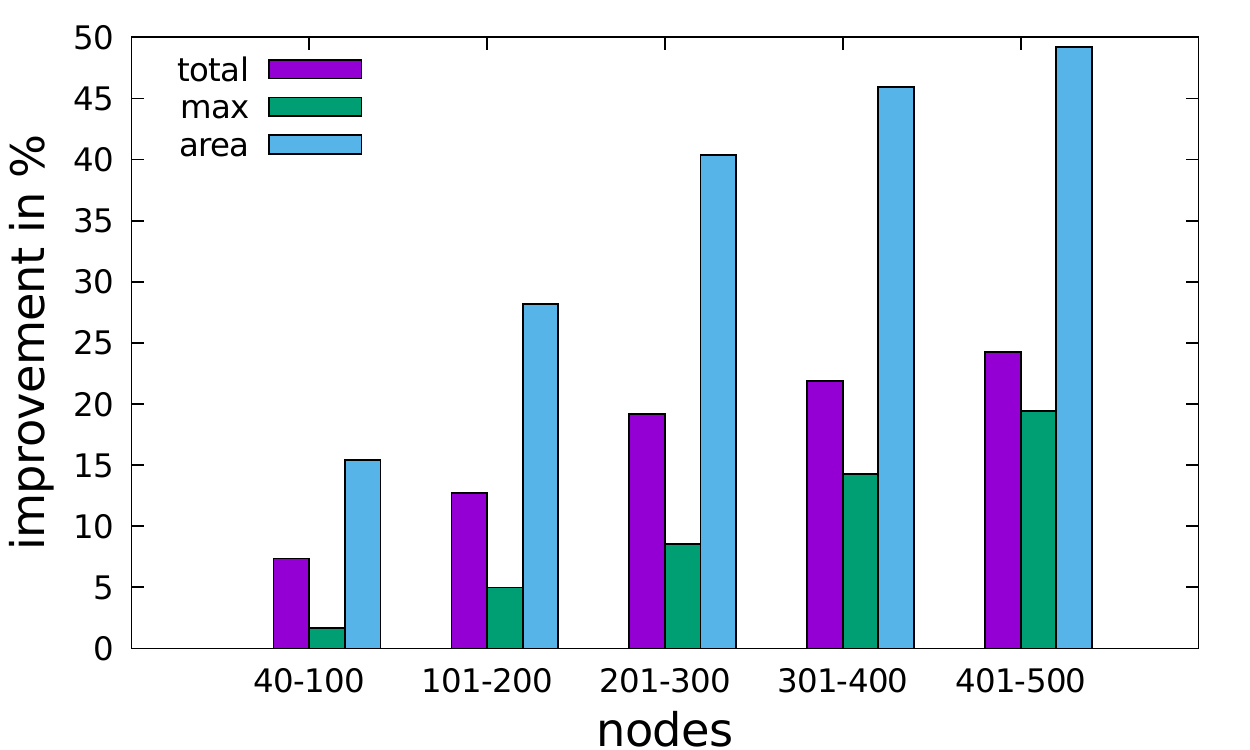}
	\end{minipage}
	\hfill
	\begin{minipage}{0.48\textwidth}
	\includegraphics[scale=0.4]{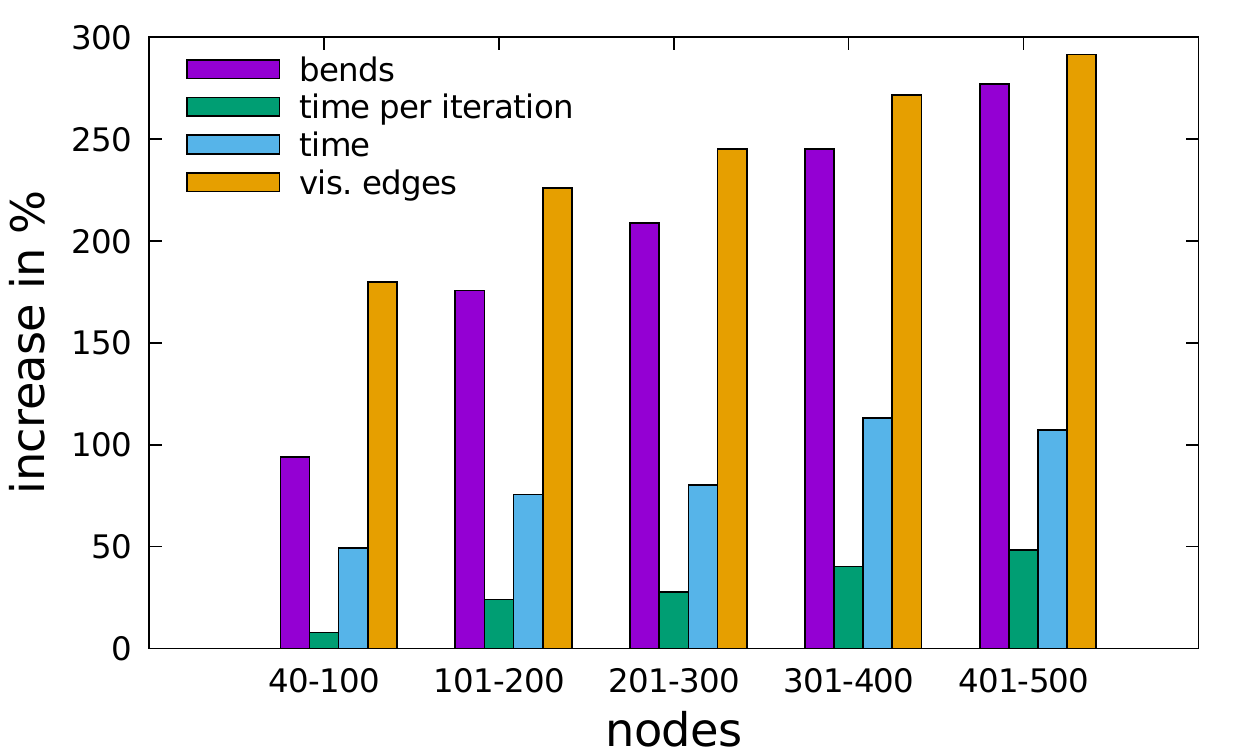}
	\end{minipage}
}%end subfloat

\caption[]{Average change of various criteria from \textsf{TRAD} to \textsf{FF} in percent for (a) the Rome graphs, (b) the quasi trees, and (c) the biconnected graphs}
\label{all_bars}
\end{figure}

We tested our algorithm on three test suites. The standard benchmark set called \emph{Rome graphs} introduced in~\cite{DIBATTISTA1997303} consists of about 11,000 real-world and real-world like graphs with 10 to 100 vertices. 
The second set of graphs are \emph{quasi-trees} which are known to be hard to compact. They have already been used in the compaction literature (e.g., \cite{Klau2001},\cite{DBLP:phd/de/Klau2004}). From this, we selected a subset of 155 graphs with 40 to 450 vertices.
The last set consists of 240 randomly generated biconnected planar 4-graphs with 40 to 500 vertices.
All used graphs were, if necessary, initially turned into planar $4$-graphs by planarizing them with methods of OGDF and replacing vertices with $k>4$ outgoing edges with faces of size $k$. This results in graph sizes of up to 314 vertices for the Rome test suite and up to 475 vertices for the quasi-trees before the orthogonalization step.
The input instances are available on \url{https://ls11-www.cs.tu-dortmund.de/mutzel/compaction}.
All tests were run on an Intel Xeon E5-2640v3 2.6GHz CPU with 128 GB RAM.

\begin{figure}[htb]
\centering
\begin{minipage}{0.48\textwidth}
	\includegraphics[scale=0.4]{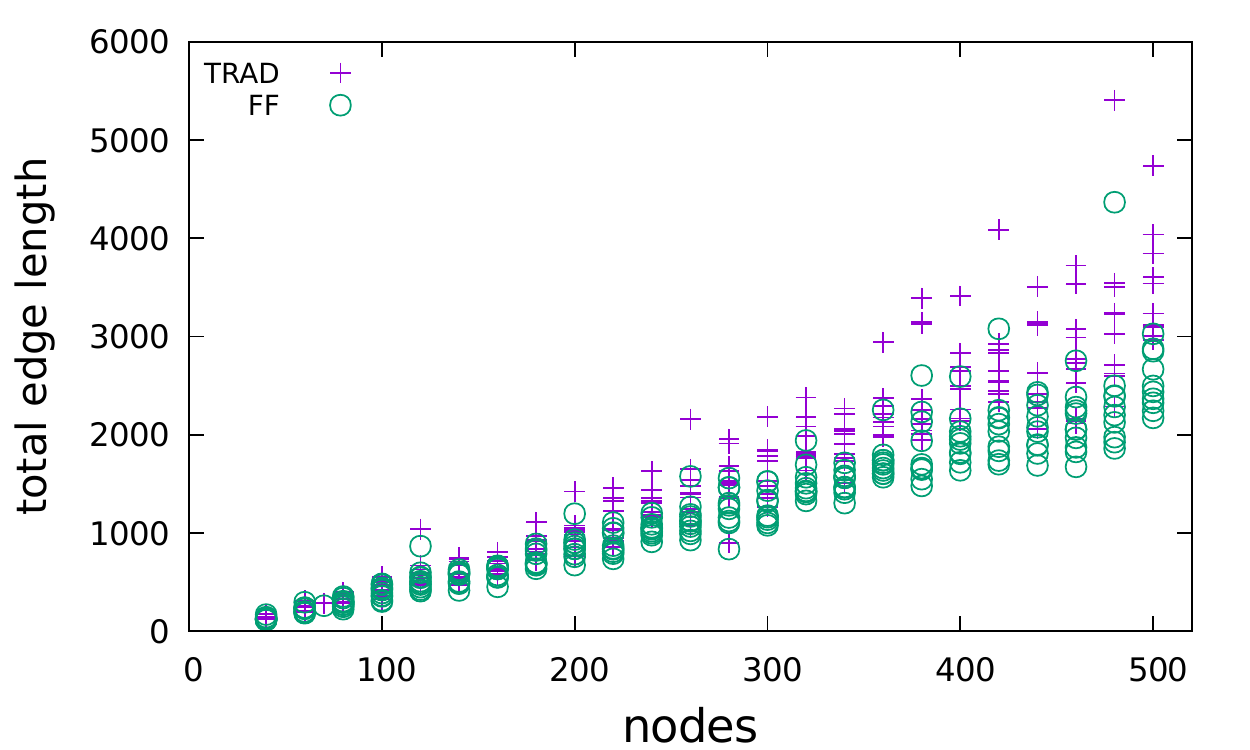}
\end{minipage}
\hfill
\begin{minipage}{0.48\textwidth}
	\includegraphics[scale=0.4]{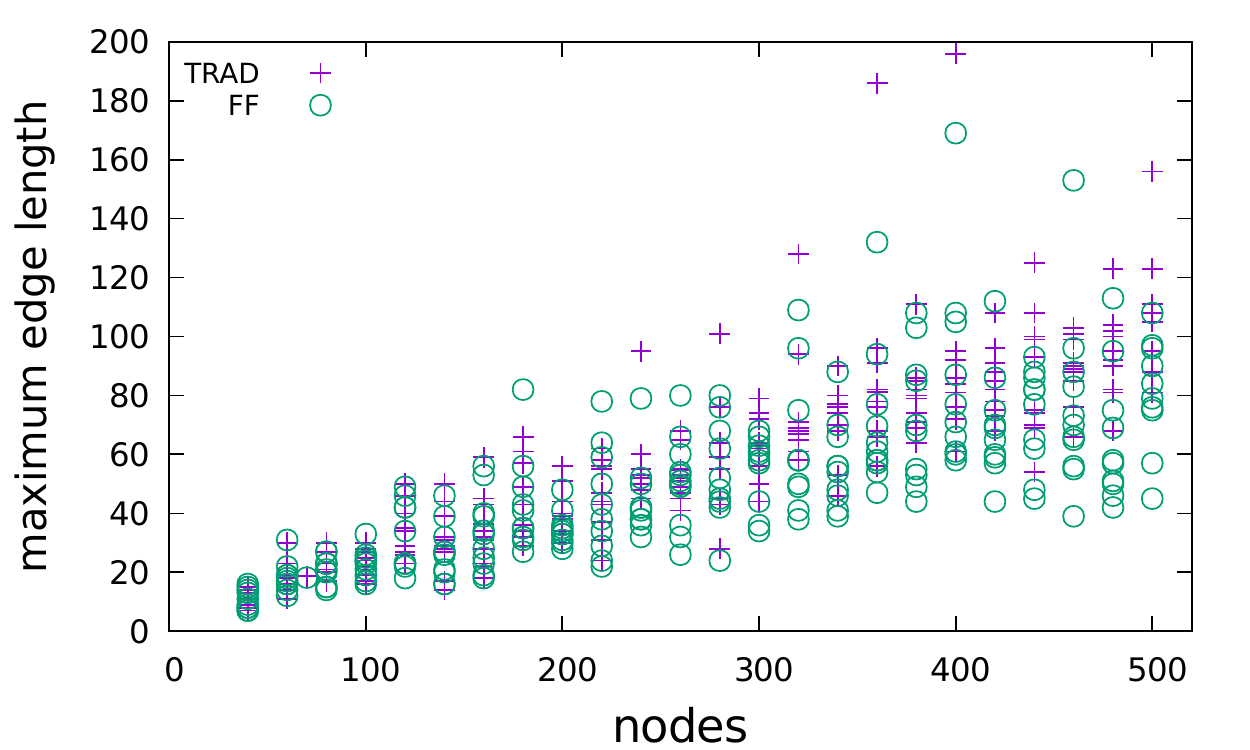}
\end{minipage}

\begin{minipage}{0.48\textwidth}
	\includegraphics[scale=0.4]{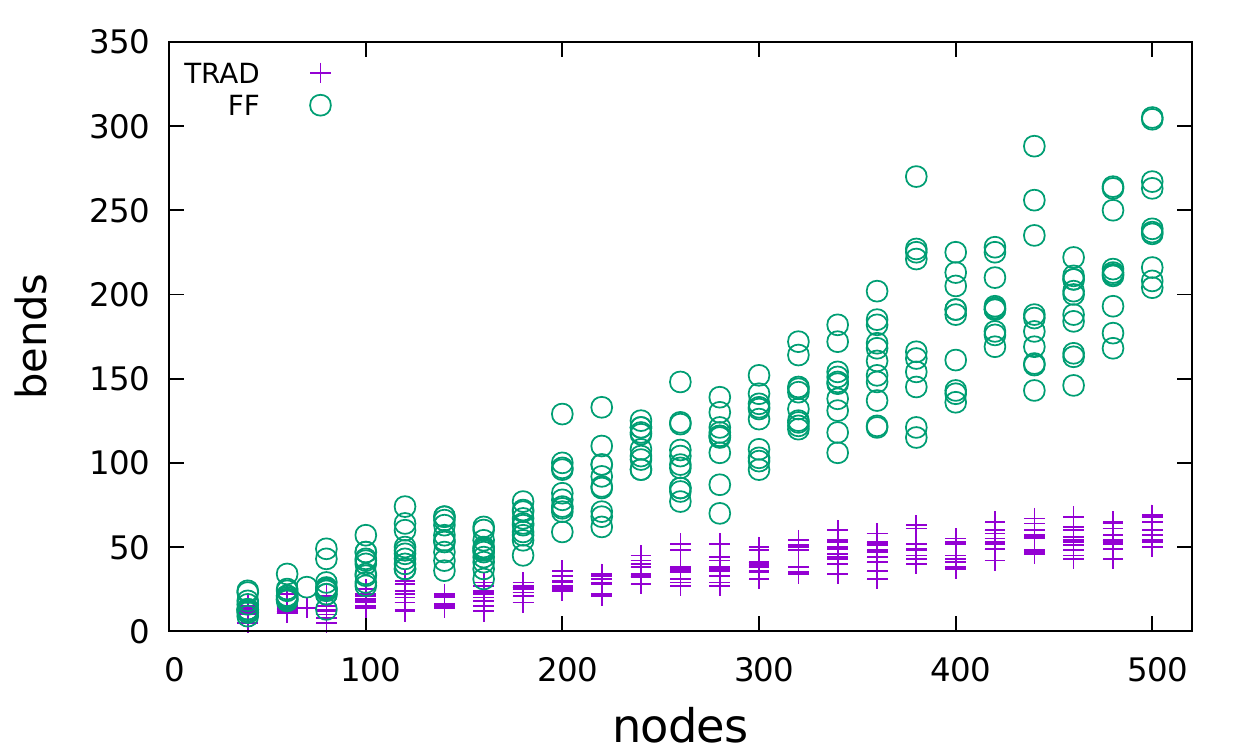}
\end{minipage}
\hfill
\begin{minipage}{0.48\textwidth}
	\includegraphics[scale=0.4]{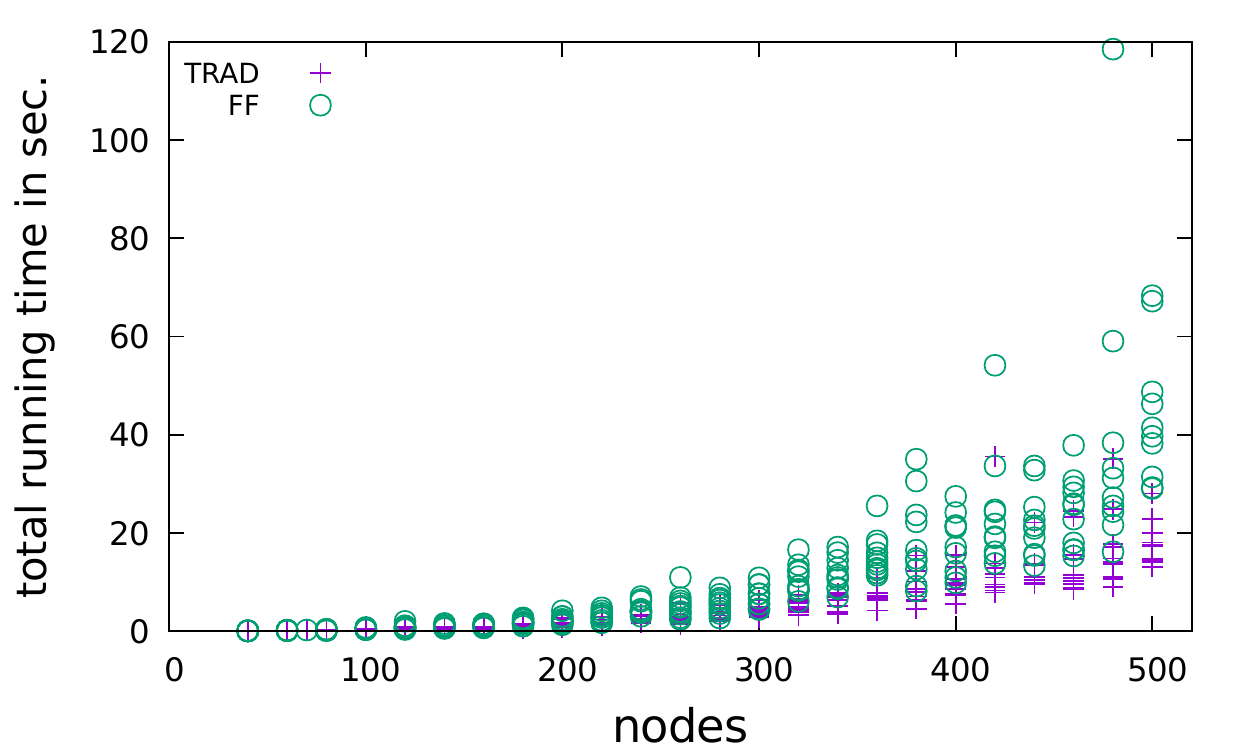}
\end{minipage}

\caption[]{Absolute results of \textsf{TRAD} and \textsf{FF} for biconnected graphs}
\label{bi_points}
\end{figure}

\noindent
\textbf{(H1)}
Figure~\ref{all_bars} shows on the left the average decrease of the total and maximum edge length as well as of the area in percent. 
In all three graph classes the total edge length as well as the area has improved. We were able to decrease the total edge length by up to 36.1\% and the area by up to 68.7\%. In general, larger graphs allow a larger improvement. 
It turned out that for the majority of the instances we were also able to reduce the maximum edge length. Only for the quasi-trees the longest edges produced by \textsf{FF} are longer on average. But in general the results are mixed, reaching from a lengthening by 87.5\% to a shortening by 58.1\%.
Figure~\ref{bi_points} shows absolute values for the total edge length and the maximum edge length measured for the biconnected graphs. (The plots for the other graph classes can be found in the Appendix.)

\medskip
\noindent
\textbf{(H2)}
The right side of Fig.~\ref{all_bars} displays the average increase of the number of bends. If an instance had no bends after \textsf{TRAD}, we use the number of bends after \textsf{FF} as relative increase to take also these instances into account for computing the average increase of bends.
As expected the drawings produced by \textsf{FF} have many more bends, especially for the Rome graphs. In one of the worst cases the number of bends went up from zero to 72. 
Figure~\ref{bi_points} shows the absolute values for the number of bends  
measured for the biconnected graphs. 
Figure~\ref{bends-area-total} gives the relation between the number of invested bends and the improvement in terms of area and total edge length for the biconneced graphs. 
The data for the other two graph classes look similar.
\begin{figure}[bth]
\subfloat[]{
\centering
\begin{minipage}{0.45\textwidth}
\includegraphics[scale=0.4]{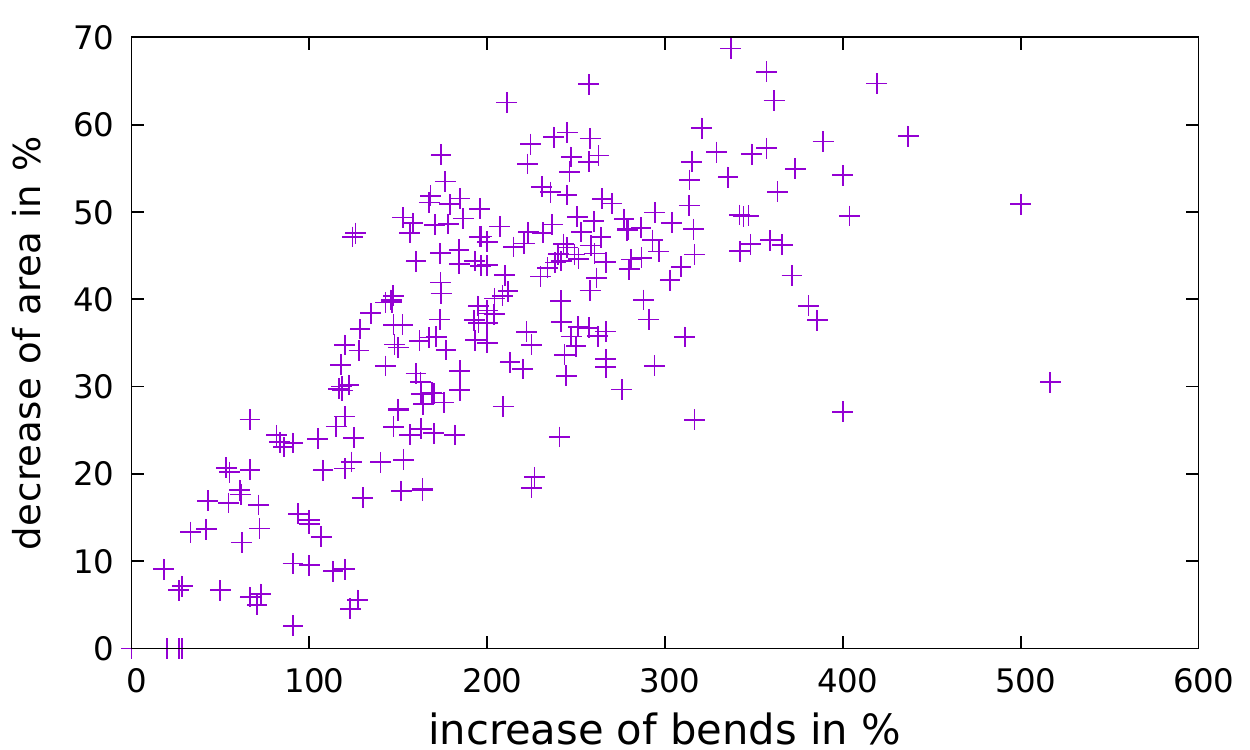}
\end{minipage}
}%end subfloat
\hfill
\subfloat[]{
\centering
\begin{minipage}{0.45\textwidth}
\includegraphics[scale=0.4]{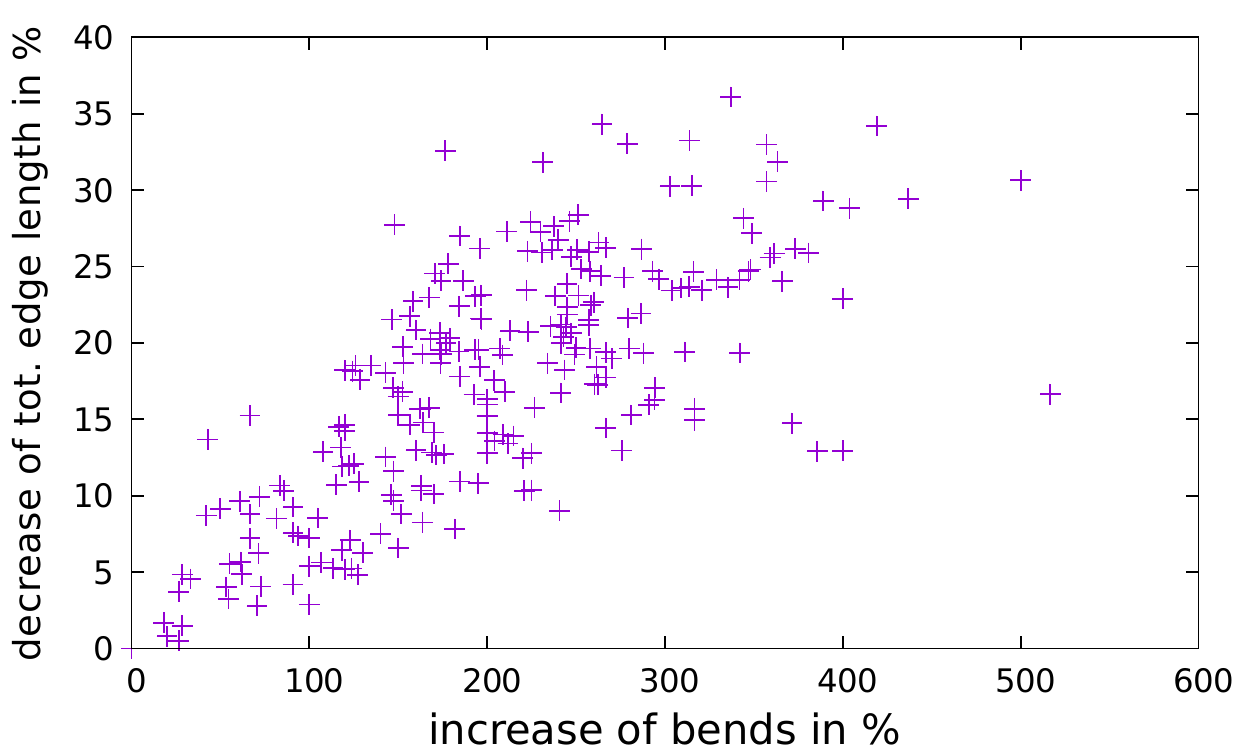}
\end{minipage}
}%end subfloat
\caption{Relation between additional bends and (a) improvement of area and (b) total edge length for the biconnected graphs}
\label{bends-area-total}
\end{figure}

\medskip
\noindent
\textbf{(H3)}
We measured the total running time and the number of performed compaction iterations. For better comparison we also display the running time per iteration, because \textsf{FF} tends to do more iterations. The reason for that is that in one compaction step of \textsf{FF} changes the drawing and therefore the flow network of the next step more significantly, leading to more new possibilities in compacting.
The right side of Fig.~\ref{all_bars} displays the average increase of the running time and the number of visibility edges. 
In all cases the number of additional dissecting edges has increased significantly. The running time per iteration has also increased, but not more than 50\% on average. Only for 252 instances of the Rome graphs and one biconnected graph the running time per iteration has more than doubled. 

\medskip
The results support the hypotheses. Our new approach is able to improve the drawing area and the total edge length by introducing additional bends.
Figure~\ref{exampledrawings} shows two examples drawn with \textsf{TRAD} and \textsf{FF}, respectively.

\begin{figure}[h]
\subfloat[][]{
\begin{minipage}{0.3\textwidth}
	\centering
	\includegraphics[scale=0.04]{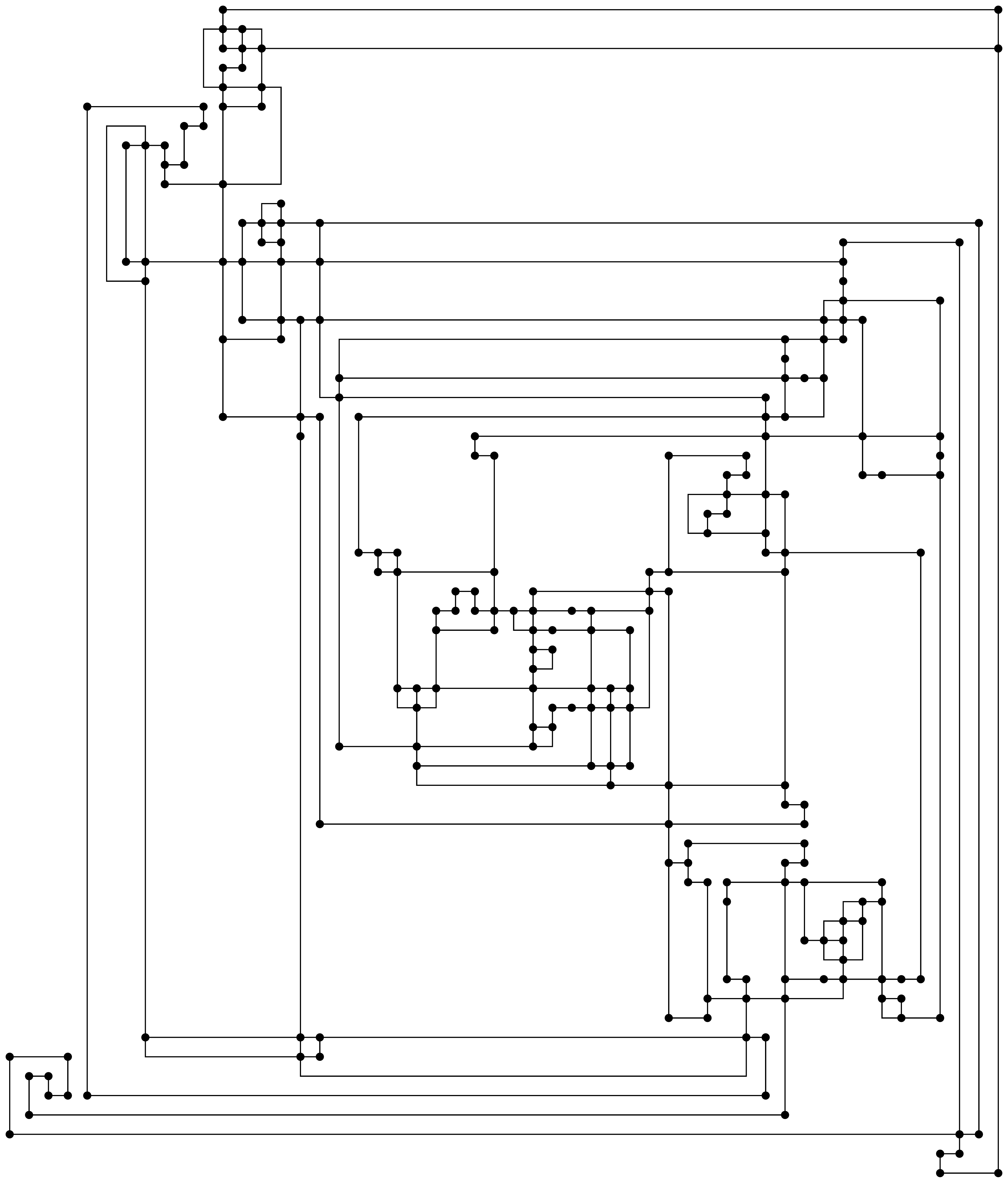}
\end{minipage}
}%end subfloat
\subfloat[][]{
\begin{minipage}{0.2\textwidth}
	\centering
	\includegraphics[scale=0.04]{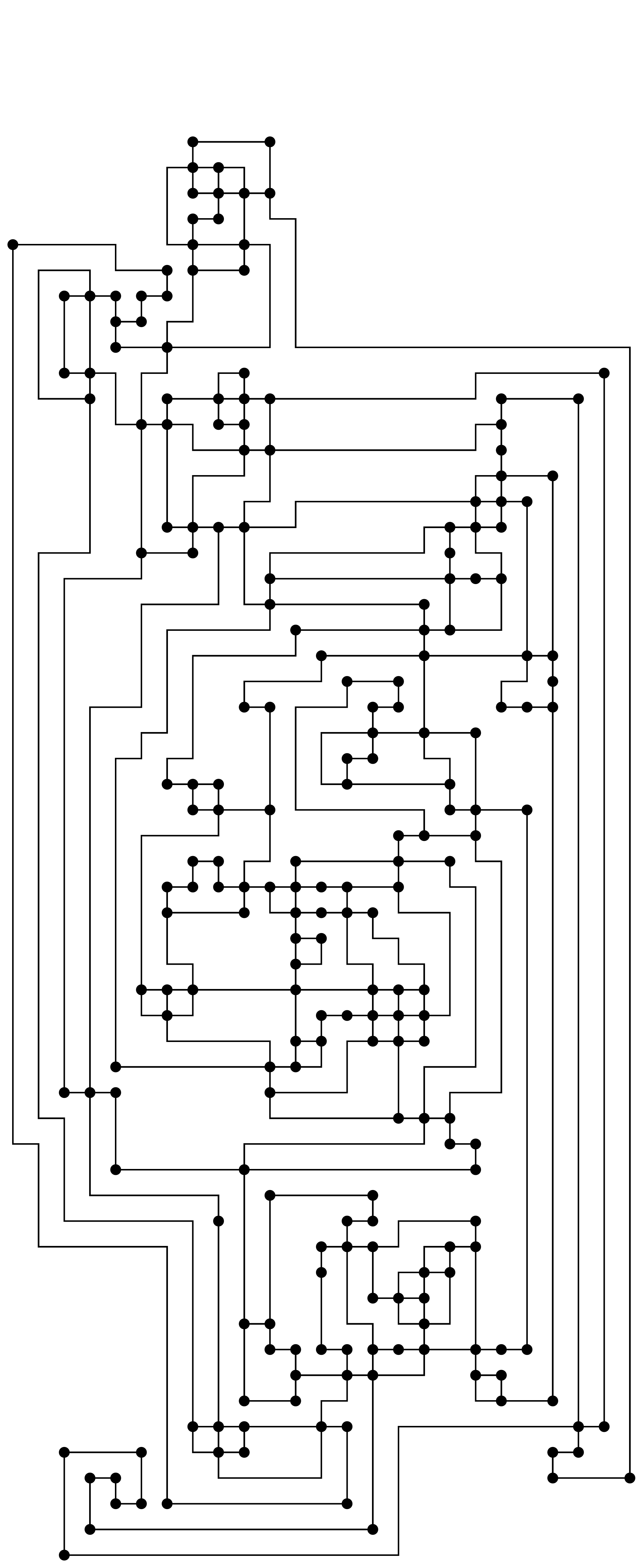}
\end{minipage}
}%end subfloat
\hfill
\subfloat[][]{
\begin{minipage}{0.2\textwidth}
	\centering
	\includegraphics[scale=0.08]{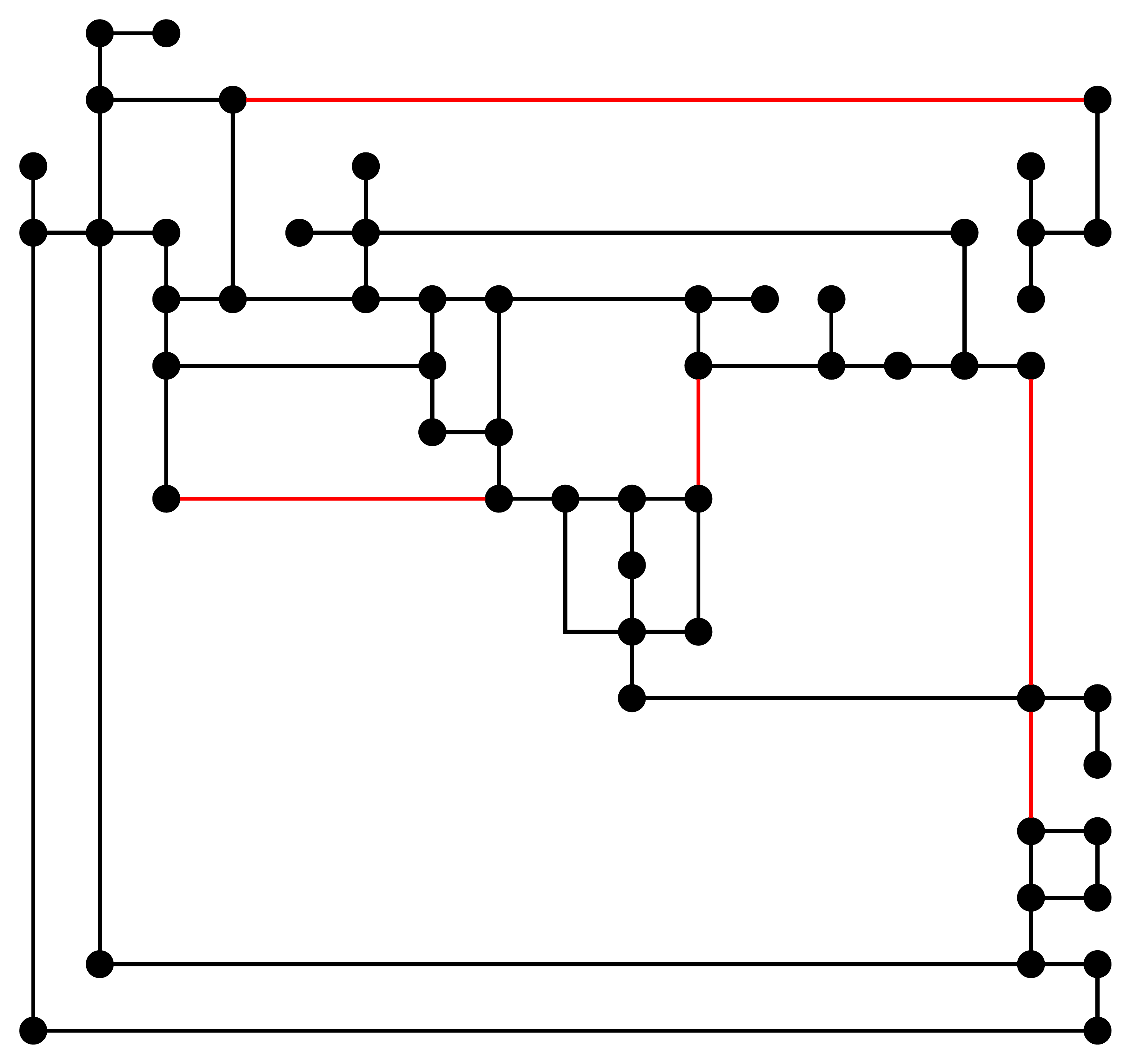}
\end{minipage}
}%end subfloat
%\hspace{2mm}
\subfloat[][]{
\begin{minipage}{0.2\textwidth}
	\centering
	\includegraphics[scale=0.08]{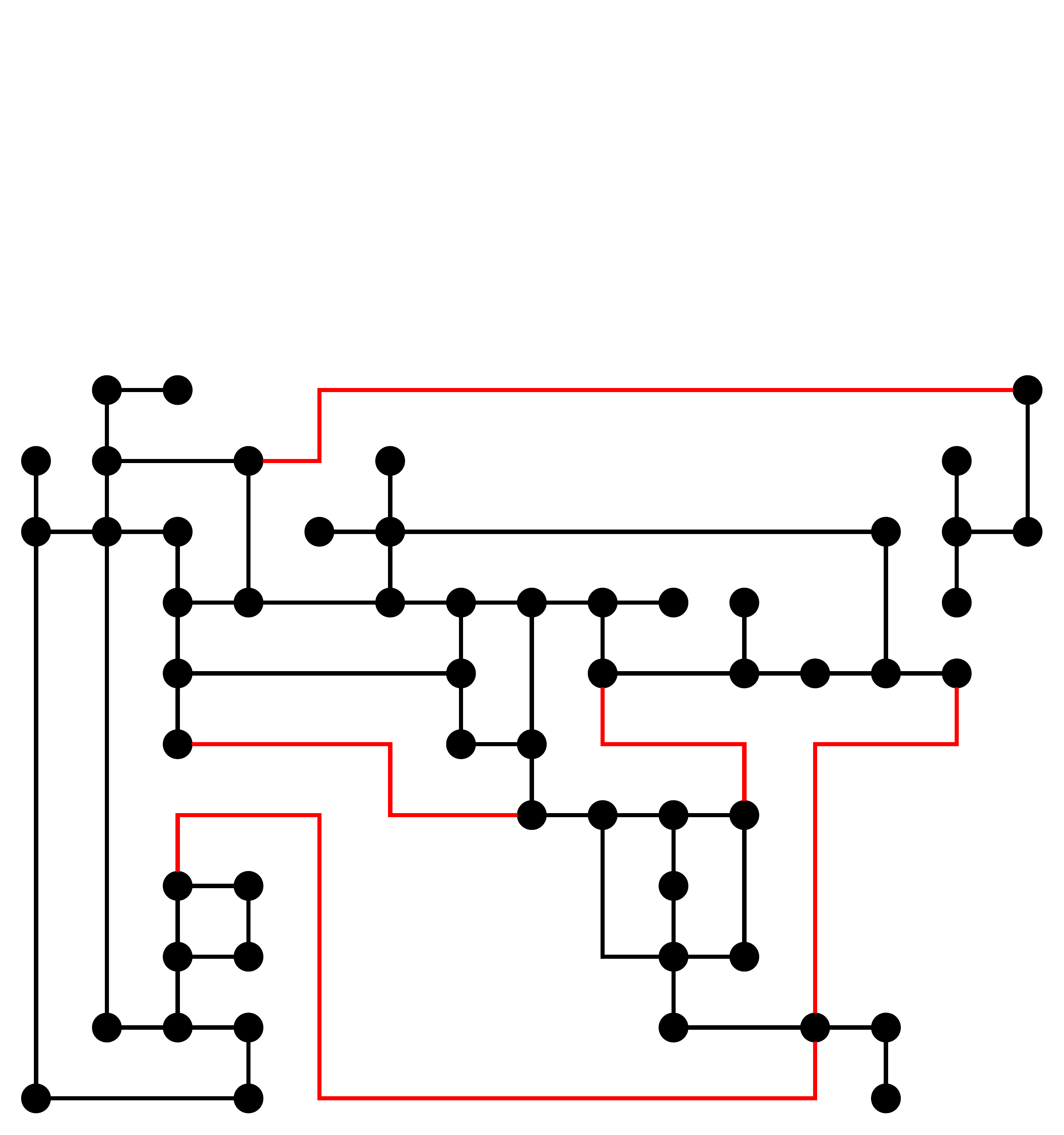}
\end{minipage}
}%end subfloat
\caption{Biconnected graph with 220 vertices and 304 edges and Rome graph with 55 vertices and 66 edges (after planarization and expanding high degree vertices): (a) \textsf{TRAD}: total edge length 1459, area 3060, bends 31 (b) \textsf{FF}: total edge length 1107, area 1320, bends 133 (c) \textsf{TRAD}: total edge length 168, area 240, bends 1 (d) \textsf{FF}: total edge length 140, area 140, bends 13. Edges with additional bends are red.}
\label{exampledrawings}
\end{figure}

\noindent\textbf{Acknowledgements.}
This work was supported by the DFG under
the project \emph{Compact Graph Drawing with Port Constraints}
(DFG MU 1129/10-1).
We would like to thank Gunnar Klau for providing us with most of our test instances.
We also gratefully acknowledge helpful discussions with Martin Gronemann, Sven Mallach, and Daniel Schmidt.

\bibliographystyle{splncs03}
\bibliography{FledFiveRestricted}

\newpage

\section*{Appendix}

\subsection*{Additional Experimental Results}

Figures~\ref{quasi_points} and \ref{rome_points} show the absolute values for the quasi trees and the Rome graphs.
The three plots in Fig.~\ref{bi_points2} complement the absolute values for the biconnected graphs of Fig.~\ref{bi_points}.
Figures~\ref{quasi_bends-area-total} and \ref{rome_bends-area-total} display the relation between additional bends and improvement of area and total edge length for the quasi trees and for the Rome graphs.

\begin{figure}[H]
\centering
\begin{minipage}{0.48 \textwidth}
	\includegraphics[scale=0.4]{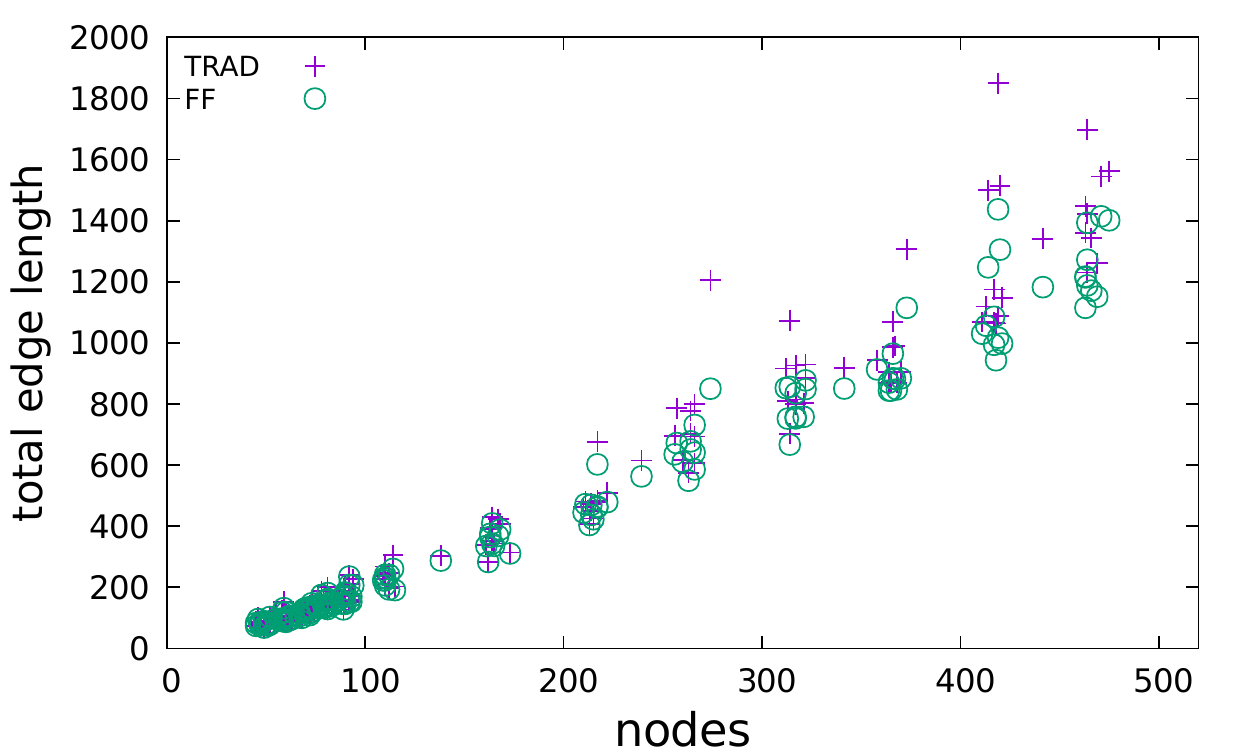}
\end{minipage}
\hfill
\begin{minipage}{0.48\textwidth}
	\includegraphics[scale=0.4]{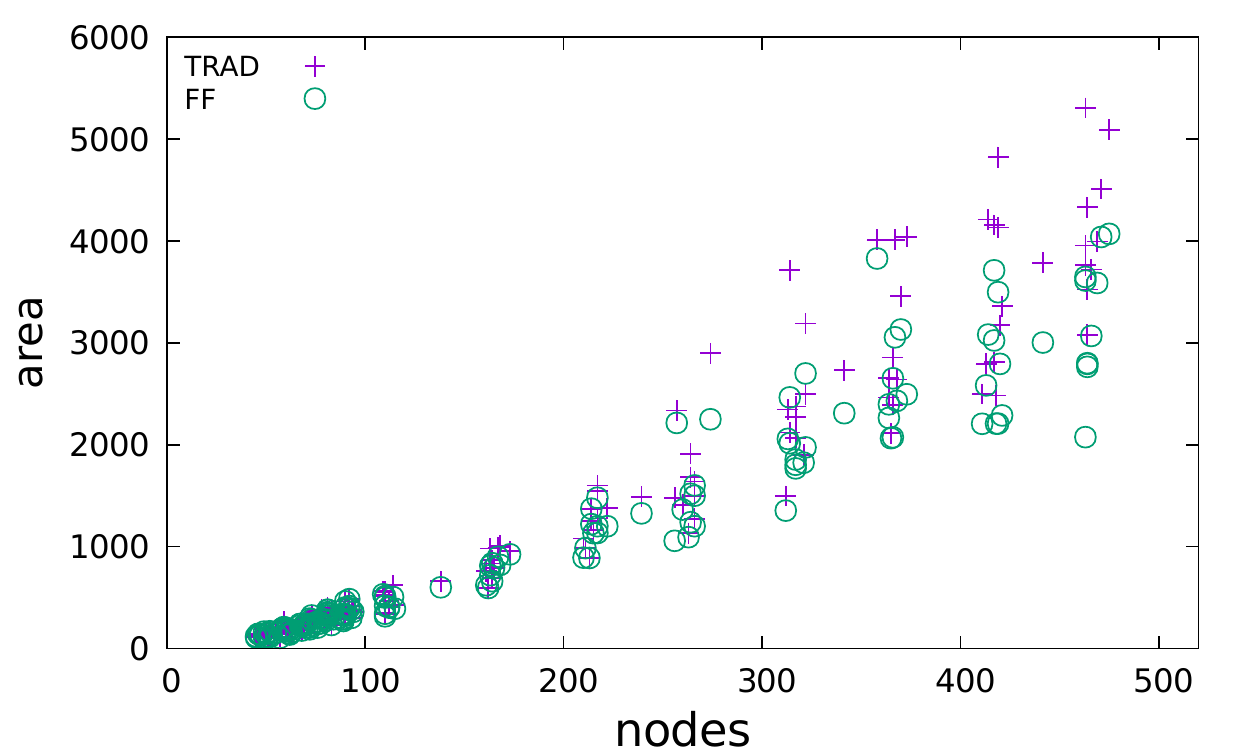}
\end{minipage}

\begin{minipage}{0.48\textwidth}
	\includegraphics[scale=0.4]{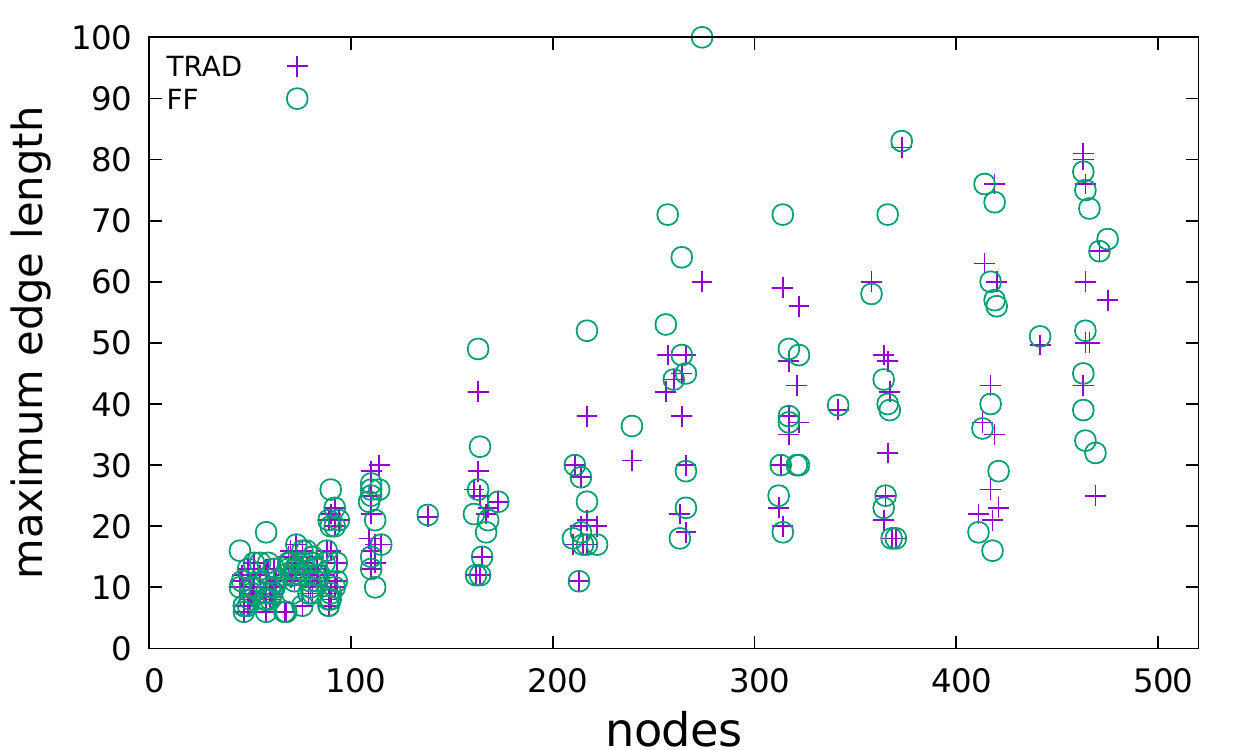}
\end{minipage}
\hfill
\begin{minipage}{0.48\textwidth}
	\includegraphics[scale=0.4]{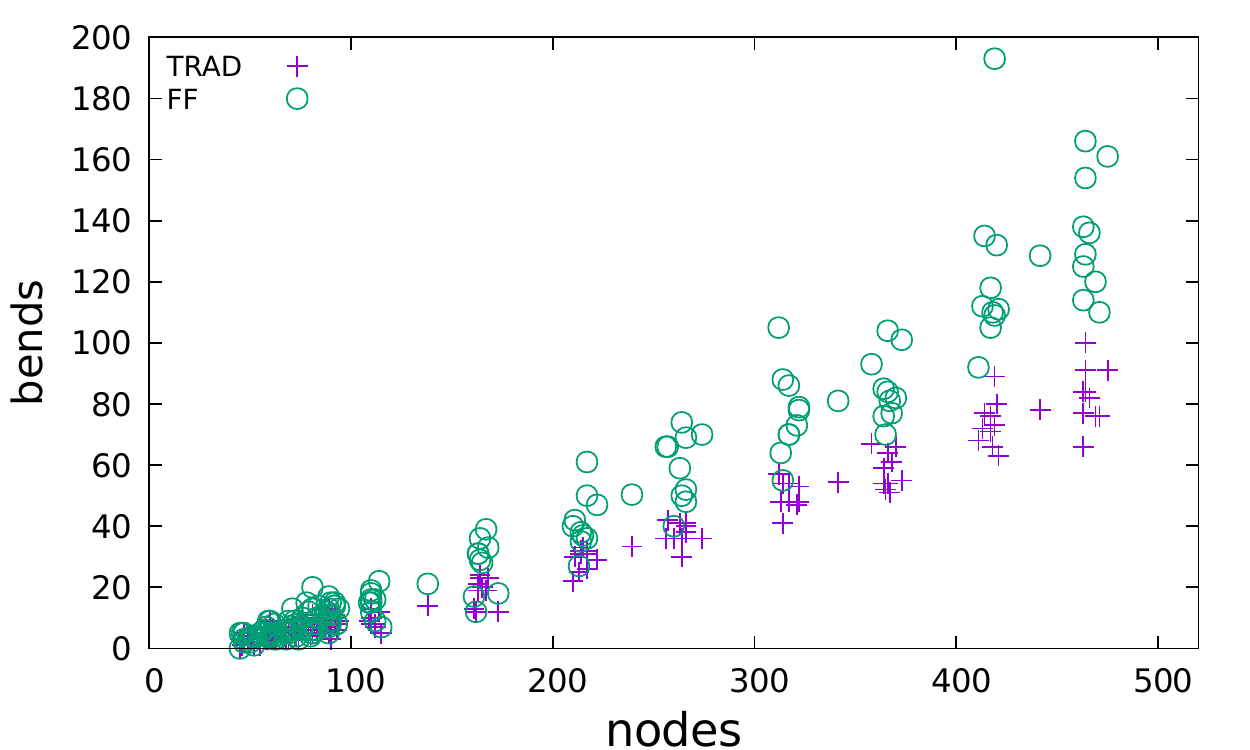}
\end{minipage}

\begin{minipage}{0.48\textwidth}
	\includegraphics[scale=0.4]{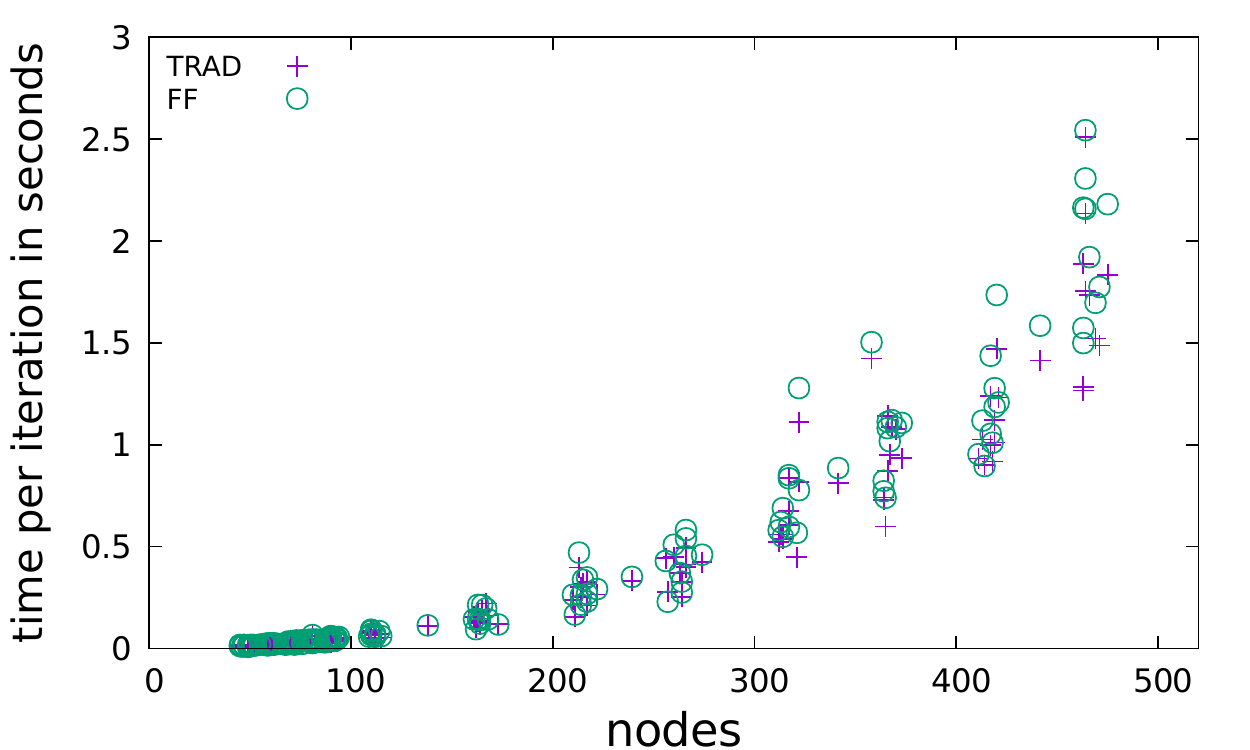}
\end{minipage}
\hfill
\begin{minipage}{0.48\textwidth}
	\includegraphics[scale=0.4]{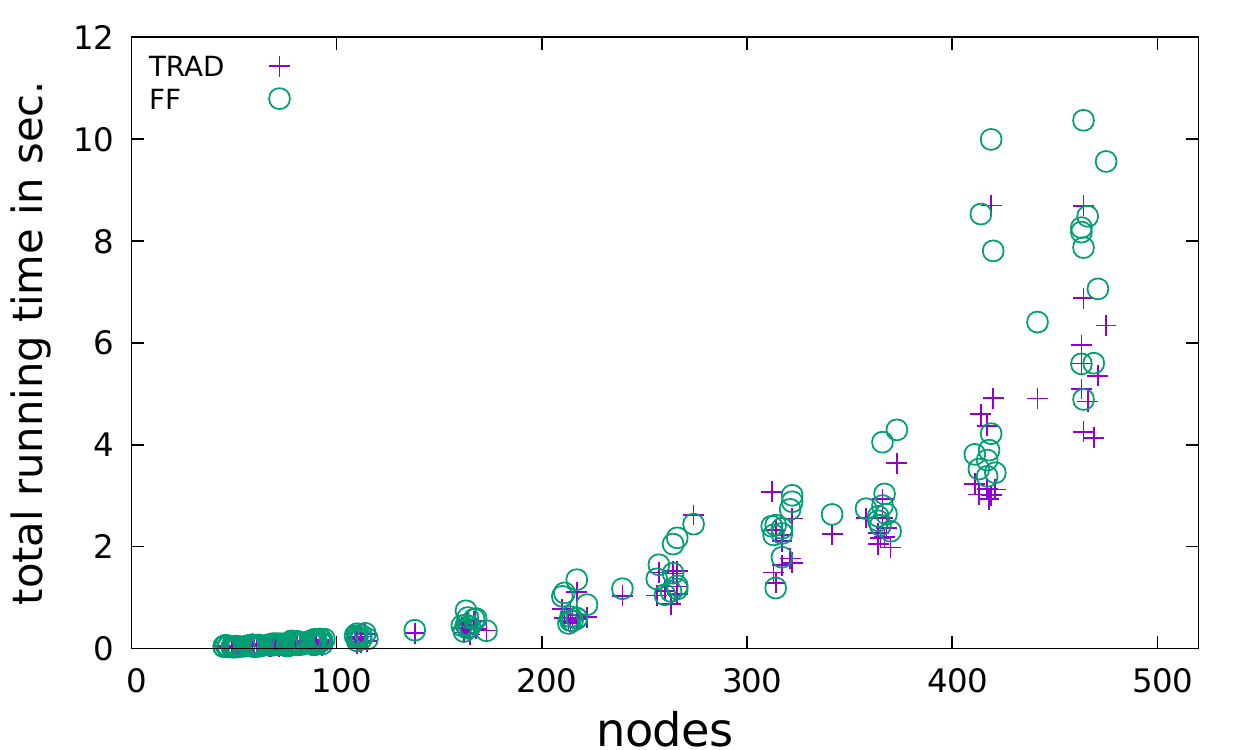}
\end{minipage}

\begin{minipage}{0.48\textwidth}
\includegraphics[scale=0.4]{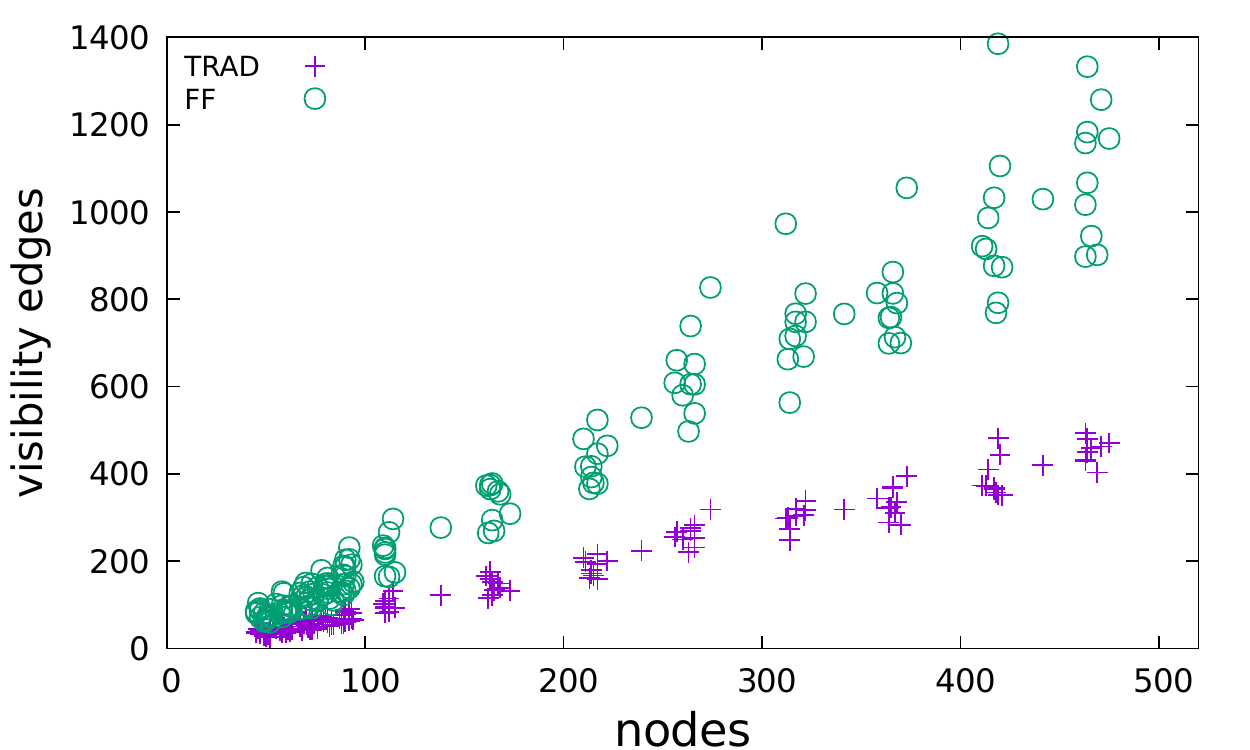}
\end{minipage}
\caption[]{Absolute results of \textsf{TRAD} and \textsf{FF} for quasi-trees}
\label{quasi_points}
\end{figure}

\begin{figure}[htb]
\centering
\begin{minipage}{0.48\textwidth}
	\includegraphics[scale=0.4]{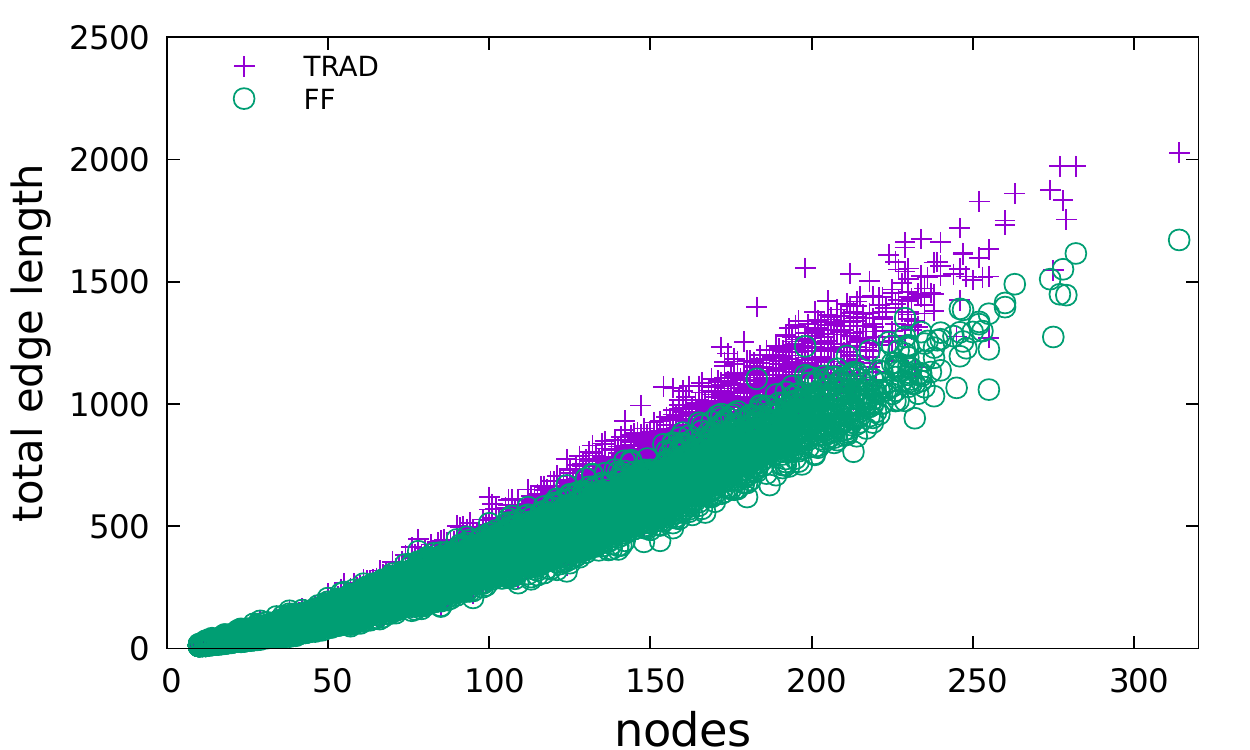}
\end{minipage}
\hfill
\begin{minipage}{0.48\textwidth}
	\includegraphics[scale=0.4]{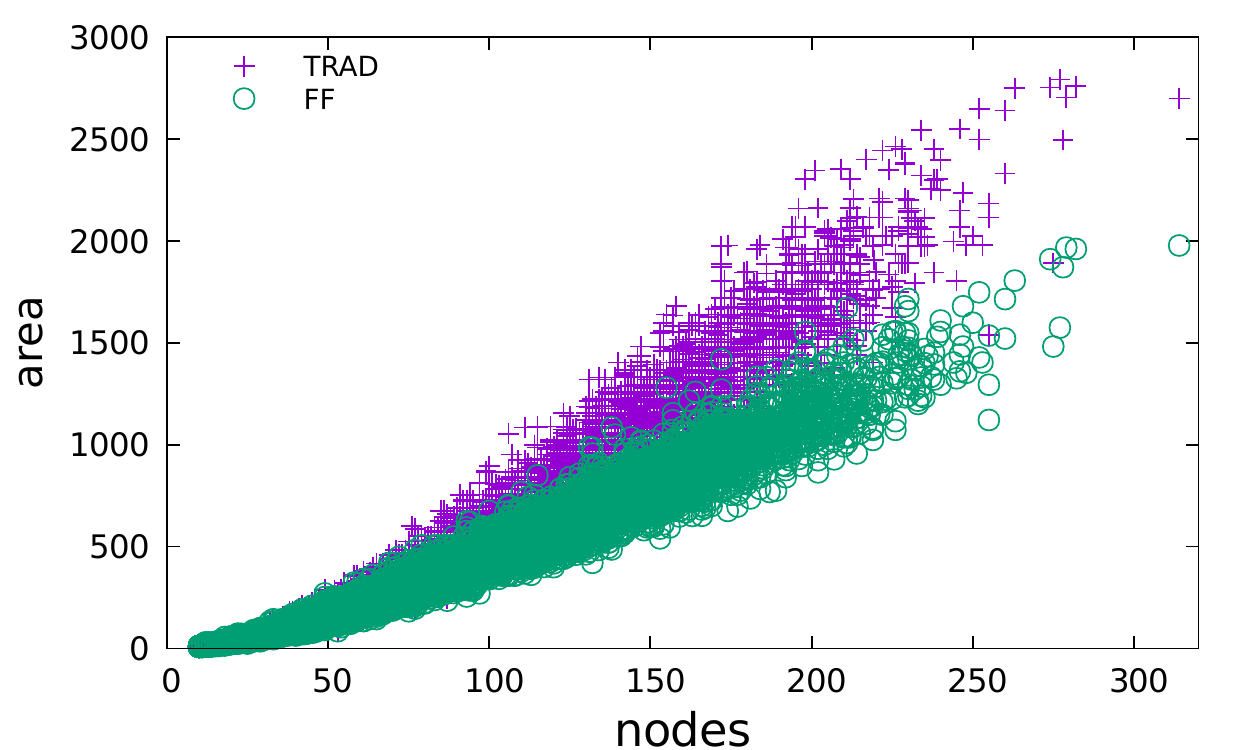}
\end{minipage}

\begin{minipage}{0.48\textwidth}
	\includegraphics[scale=0.4]{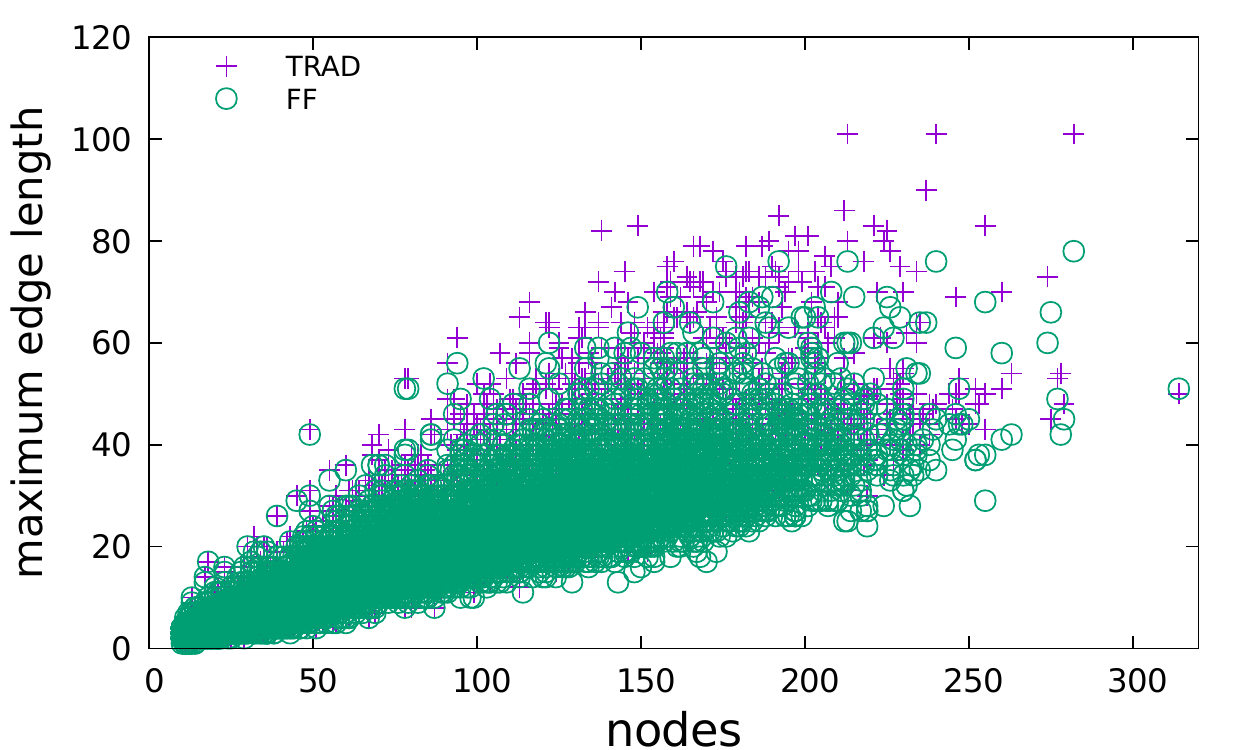}
\end{minipage}
\hfill
\begin{minipage}{0.48\textwidth}
	\includegraphics[scale=0.4]{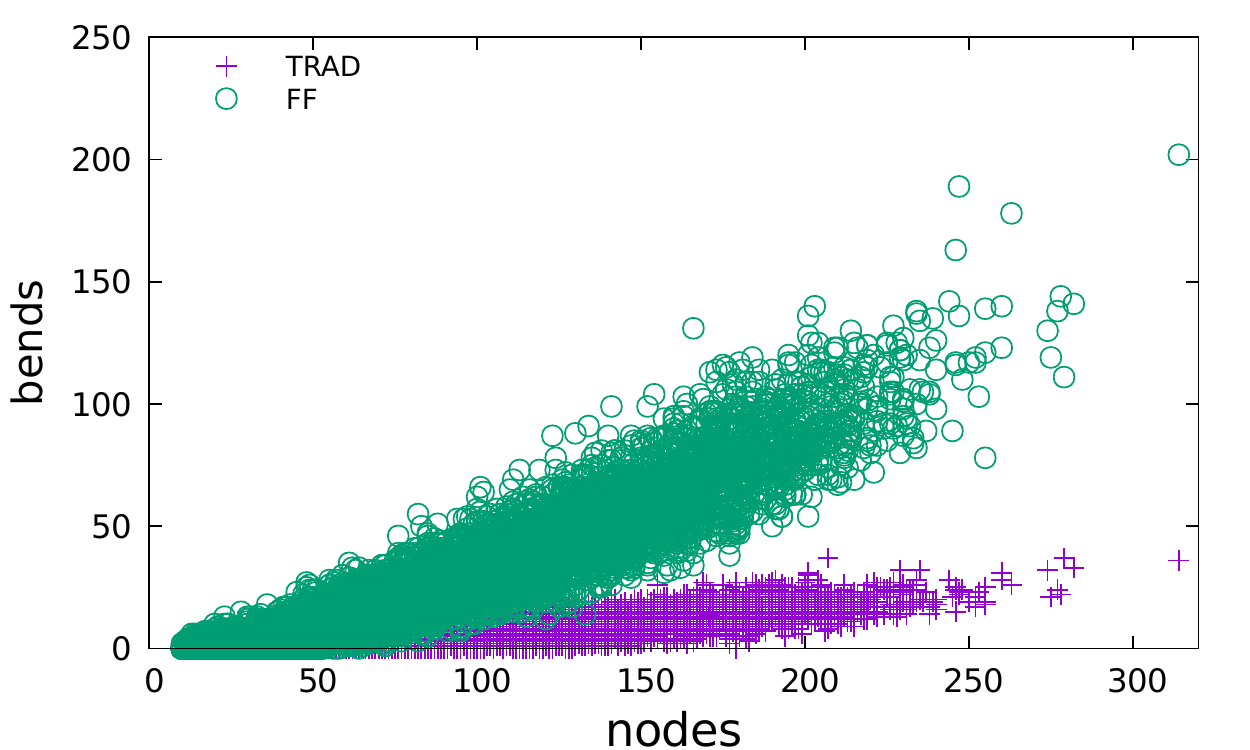}
\end{minipage}

\begin{minipage}{0.48\textwidth}
	\includegraphics[scale=0.4]{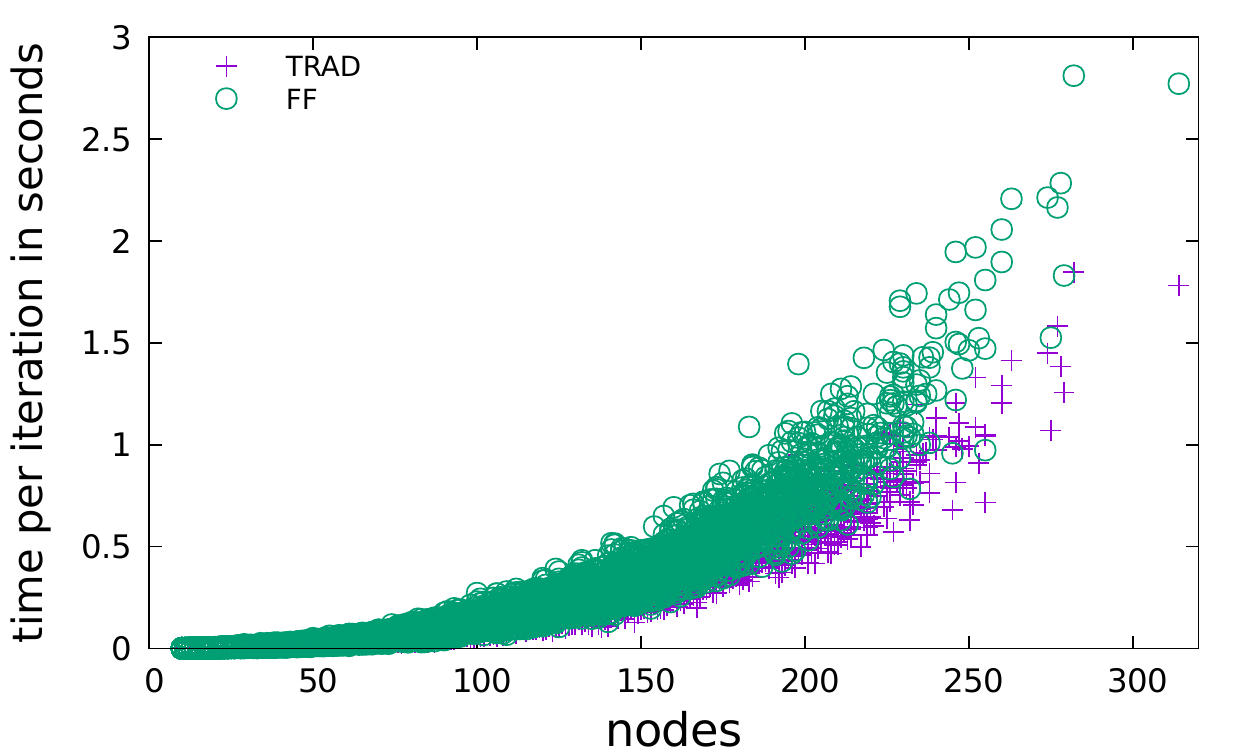}
\end{minipage}
\hfill
\begin{minipage}{0.48\textwidth}
	\includegraphics[scale=0.4]{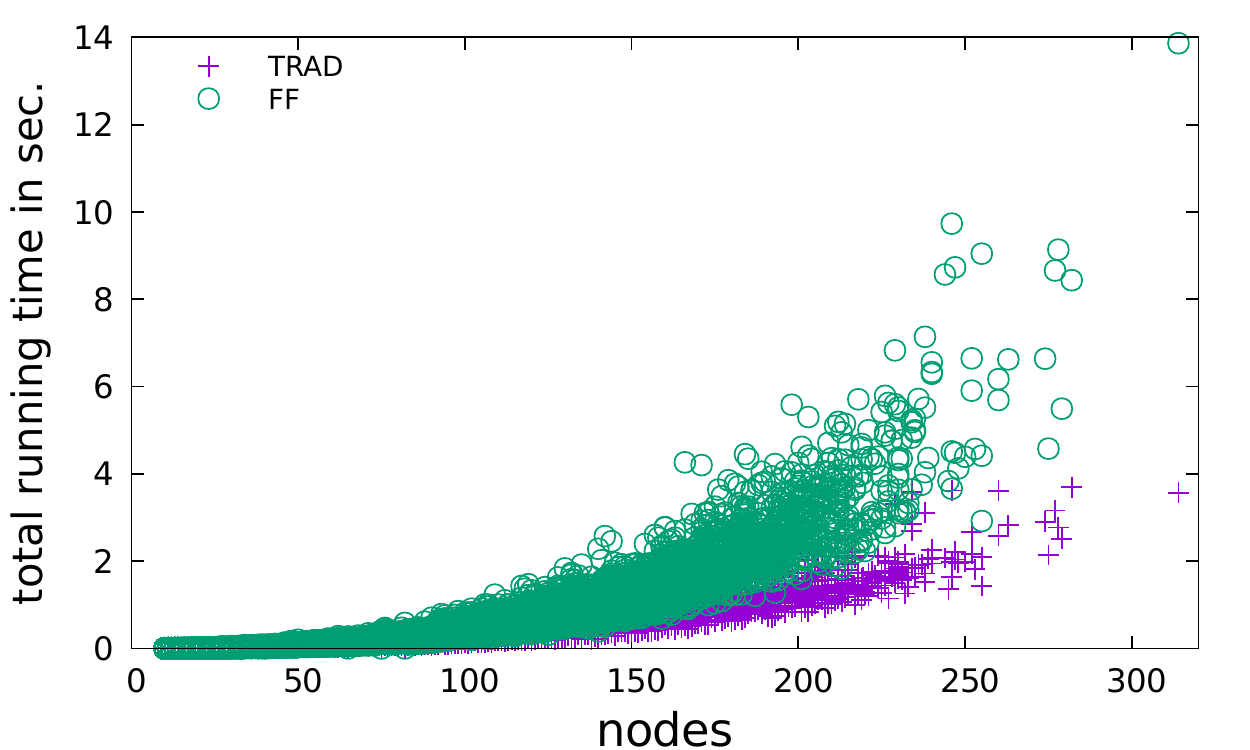}
\end{minipage}

\begin{minipage}{0.48\textwidth}
\includegraphics[scale=0.4]{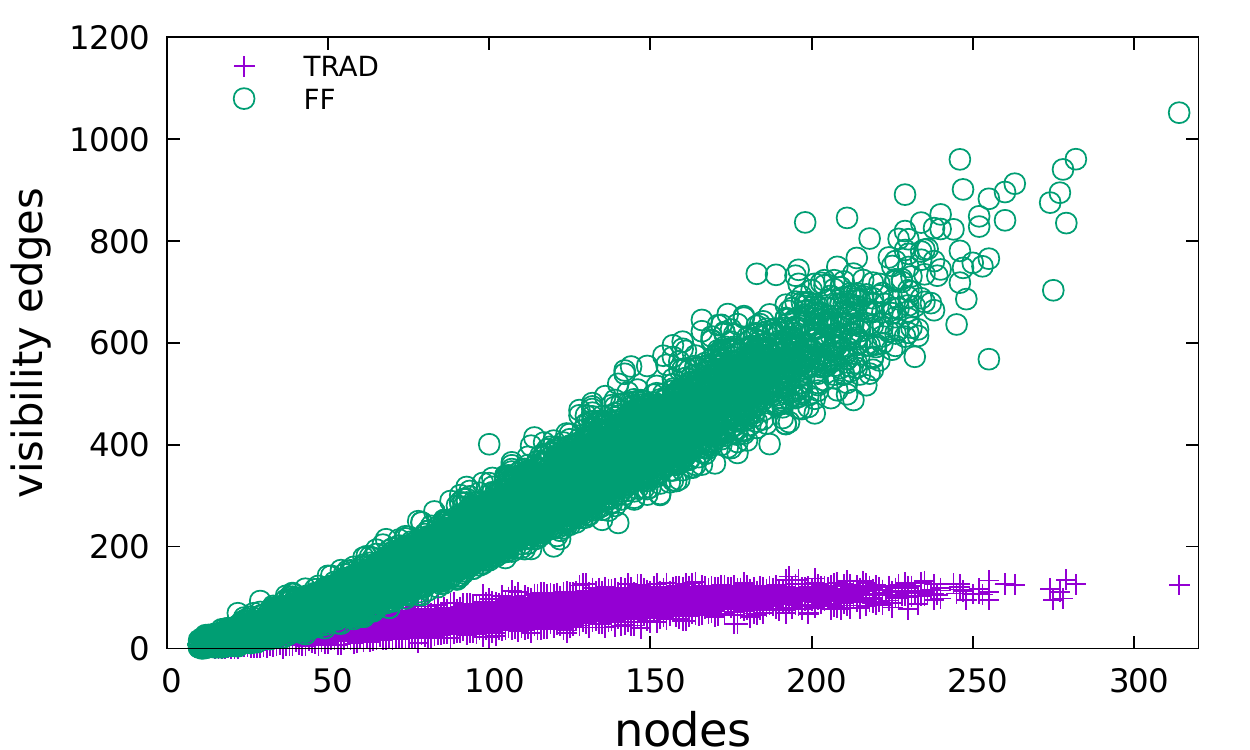}
\end{minipage}
\caption[]{Absolute results of \textsf{TRAD} and \textsf{FF} for Rome graphs}
\label{rome_points}
\end{figure}

\begin{figure}[htb]
\centering

\begin{minipage}{0.48\textwidth}
	\includegraphics[scale=0.4]{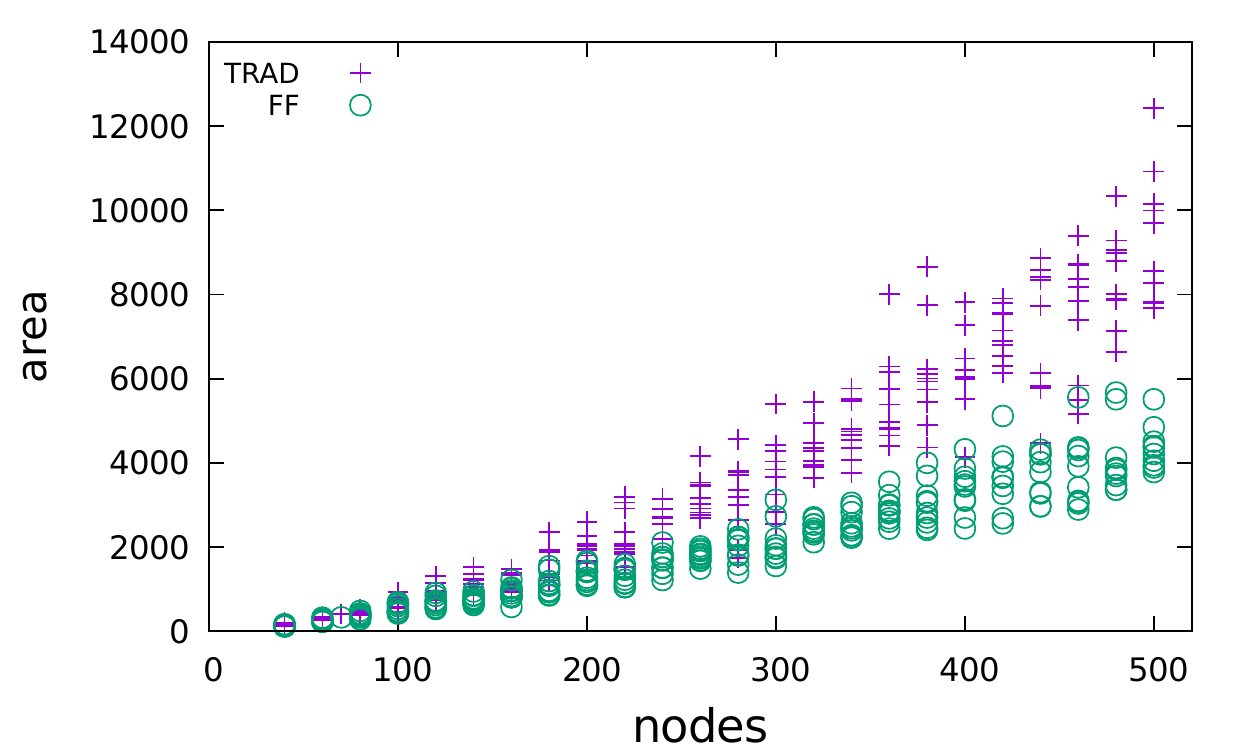}
\end{minipage}
\hfill
\begin{minipage}{0.48\textwidth}
	\includegraphics[scale=0.4]{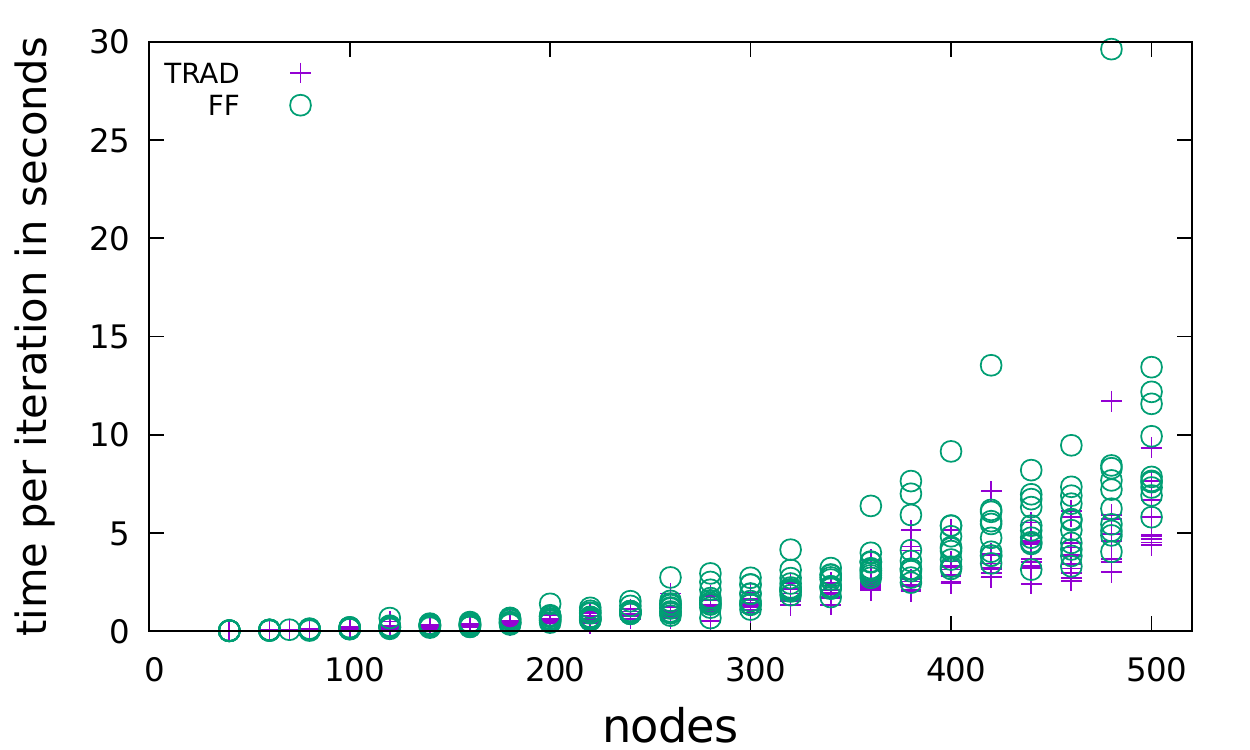}
\end{minipage}

\begin{minipage}{0.48\textwidth}
\includegraphics[scale=0.4]{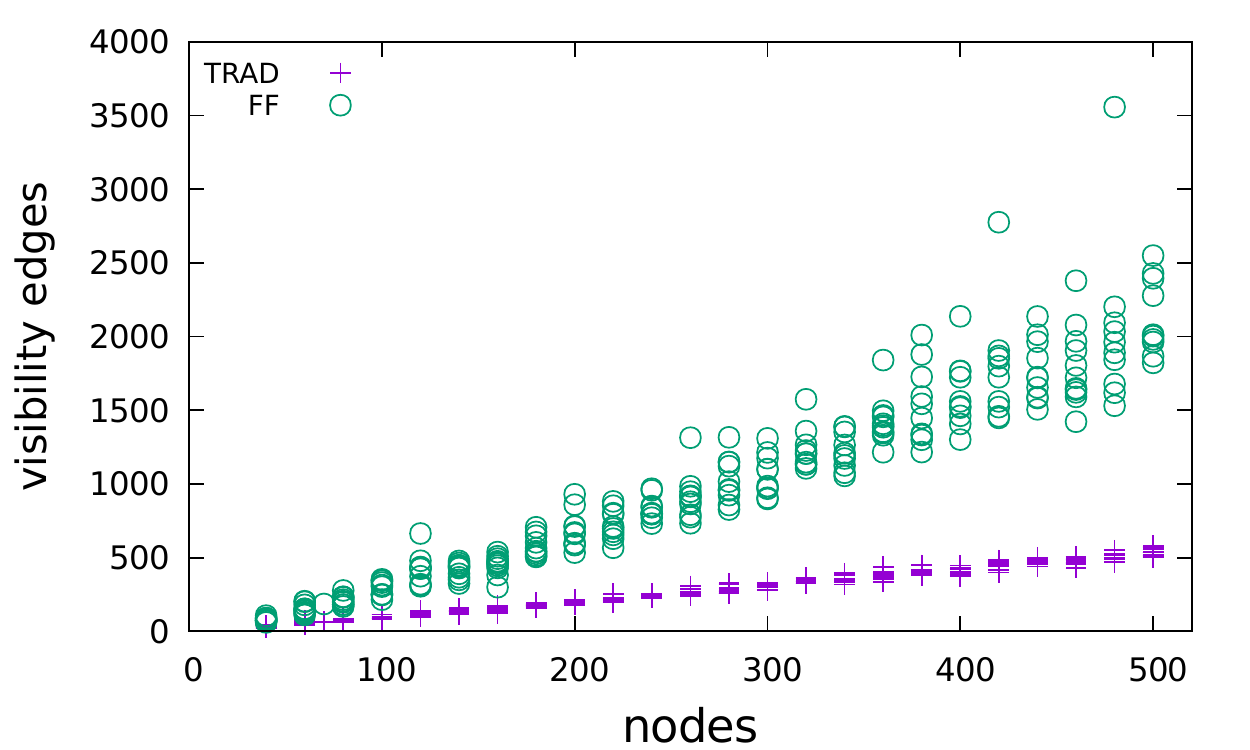}
\end{minipage}
\caption[]{Remaining absolute results of \textsf{TRAD} and \textsf{FF} for biconnected graphs}
\label{bi_points2}
\end{figure}

\begin{figure}[htb]
\begin{minipage}{0.48\textwidth}
\includegraphics[scale=0.4]{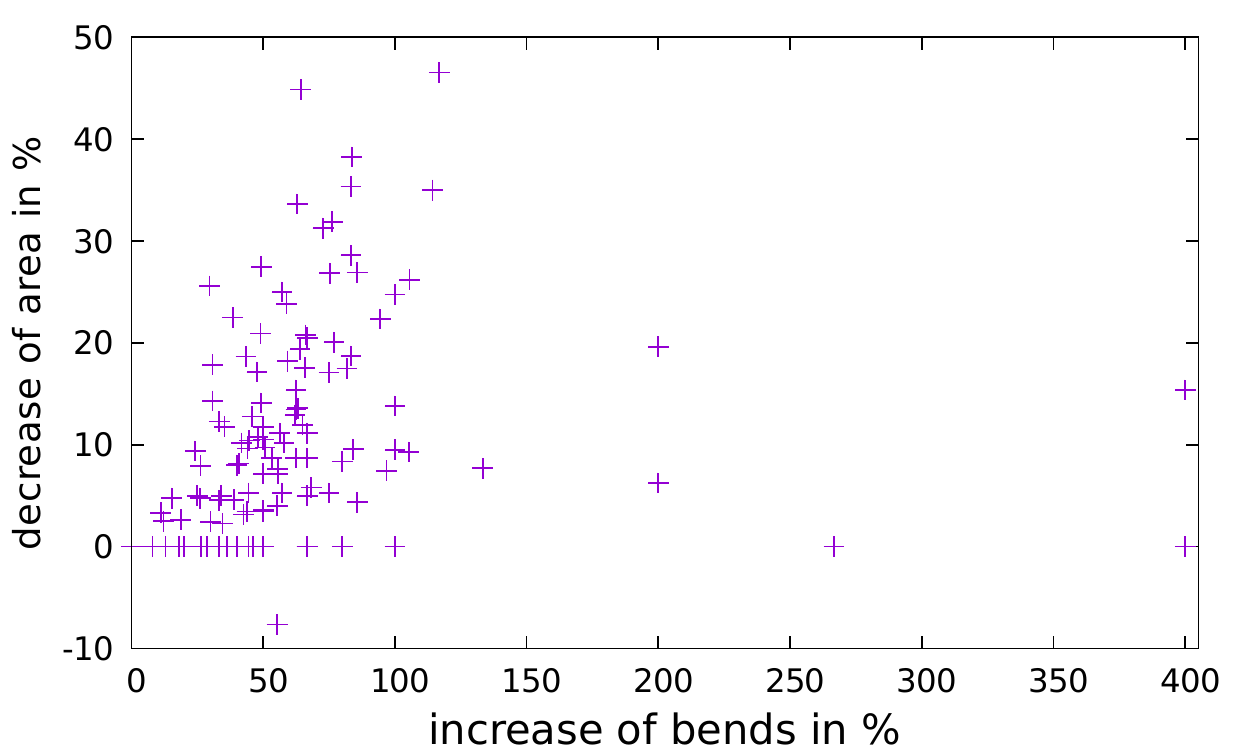}
\end{minipage}
\hfill
\begin{minipage}{0.48 \textwidth}
\includegraphics[scale=0.4]{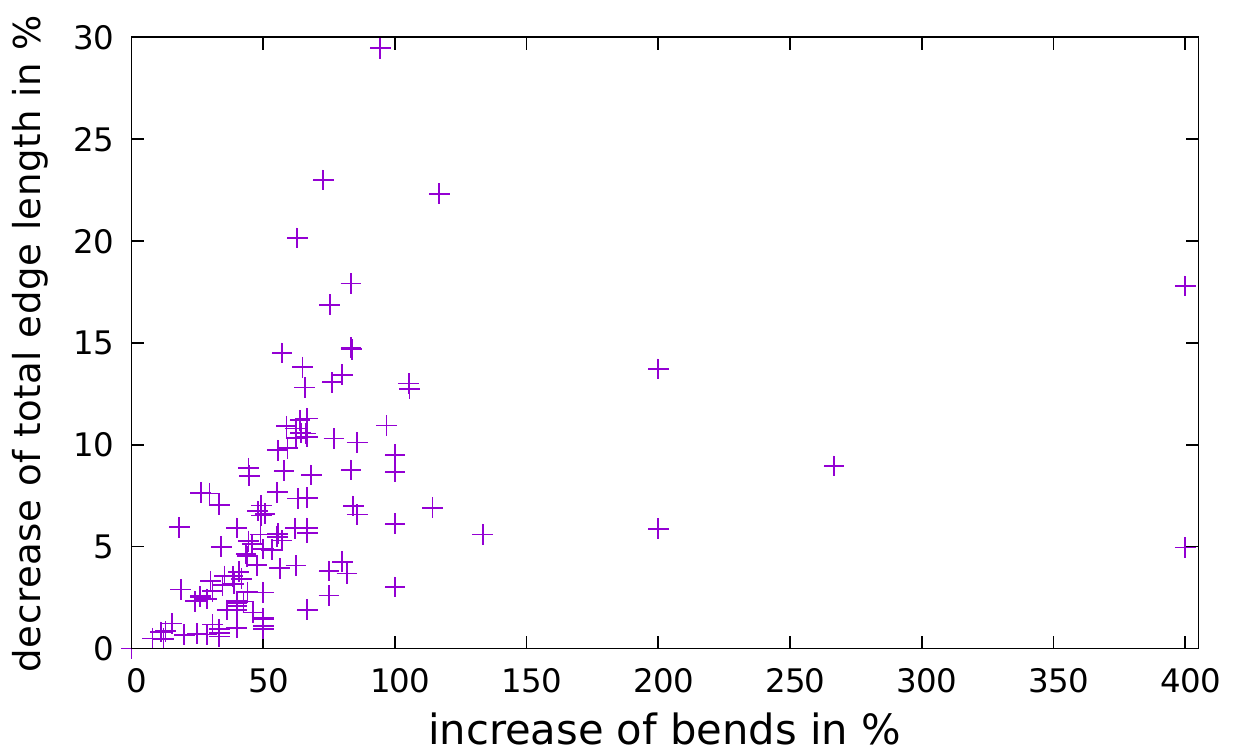}
\end{minipage}
\caption{Relation between additional bends and improvement of area and total edge length for the quasi trees}
\label{quasi_bends-area-total}
\end{figure}

\begin{figure}[htb]
\begin{minipage}{0.48\textwidth}
\includegraphics[scale=0.4]{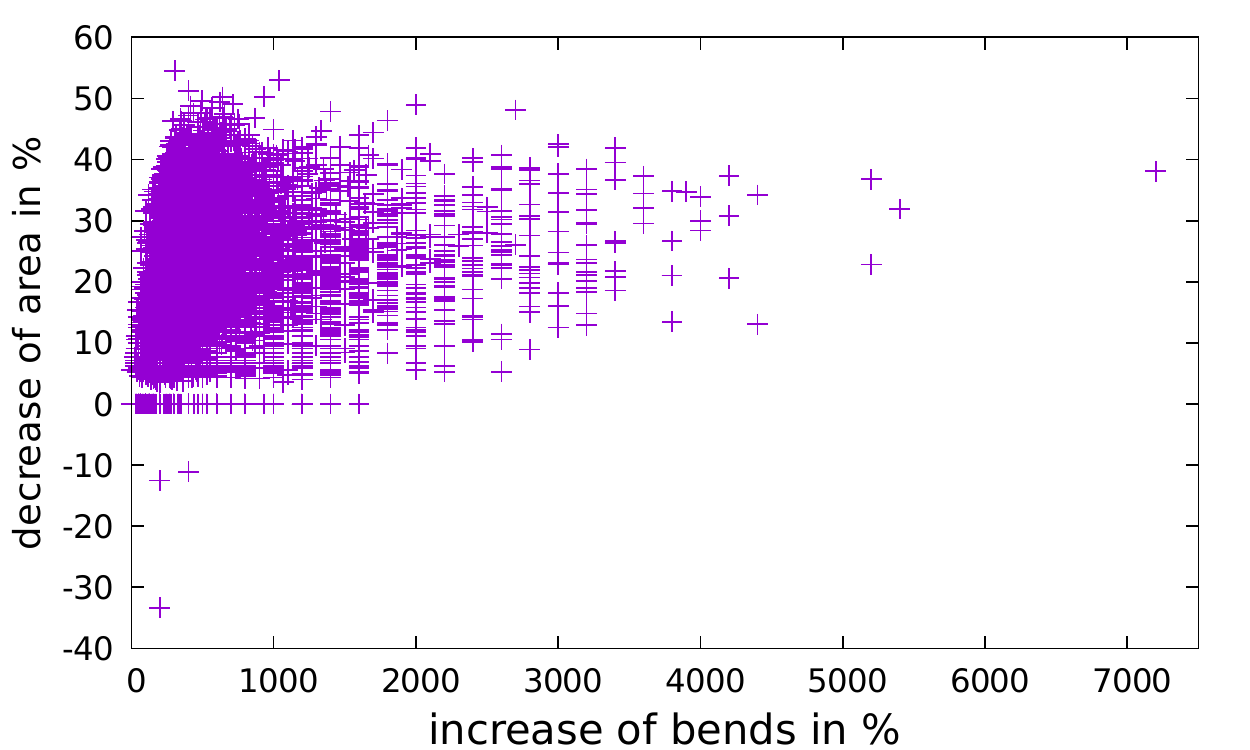}
\end{minipage}
\hfill
\begin{minipage}{0.48\textwidth}
\includegraphics[scale=0.4]{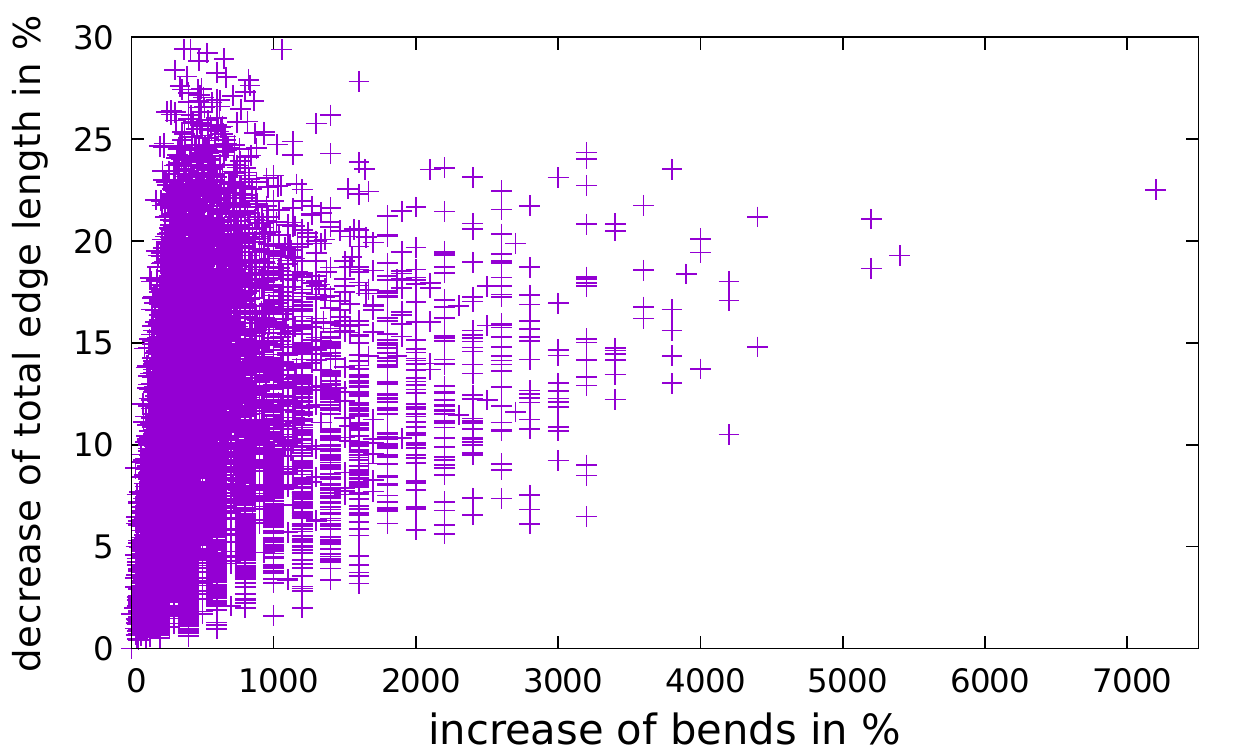}
\end{minipage}
\caption{Relation between additional bends and improvement of area and total edge length for the Rome graphs}
\label{rome_bends-area-total}
\end{figure}

\end{document}